\documentclass[11pt,leqno]{amsart}
\usepackage{amsthm,amsfonts,amssymb,amsmath,oldgerm}
\numberwithin{equation}{section}
\usepackage{fullpage}
\usepackage{color}
\usepackage{calrsfs}
\DeclareMathAlphabet{\pazocal}{OMS}{zplm}{m}{n}
\usepackage{graphicx, xcolor}
\usepackage{subcaption}

\graphicspath{{draws/}}

\newcommand\R{\mathbb R}

\newcommand\br{\begin{remark}}
\newcommand\er{\end{remark}}
\newcommand\bp{\begin{pmatrix}}
\newcommand\ep{\end{pmatrix}}
\newcommand{\be}{\begin{equation}}
\newcommand{\ee}{\end{equation}}
\newcommand\ba{\begin{equation}\begin{aligned}}
\newcommand\ea{\end{aligned}\end{equation}}

\newcommand{\bap}{\begin{app}}
\newcommand{\eap}{\end{app}}
\newcommand{\begs}{\begin{exams}}
\newcommand{\eegs}{\end{exams}}
\newcommand{\beg}{\begin{example}}
\newcommand{\eeg}{\end{exaplem}}
\newcommand{\bpr}{\begin{proposition}}
\newcommand{\epr}{\end{proposition}}
\newcommand{\bt}{\begin{theorem}}
\newcommand{\et}{\end{theorem}}
\newcommand{\bc}{\begin{corollary}}
\newcommand{\ec}{\end{corollary}}
\newcommand{\bl}{\begin{lemma}}
\newcommand{\el}{\end{lemma}}
\newcommand{\bd}{\begin{definition}}
\newcommand{\ed}{\end{definition}}
\newcommand{\brs}{\begin{remarks}}
\newcommand{\ers}{\end{remarks}}

\newcommand{\mycomment}[1]{}

\newtheorem{theorem}{Theorem}[section]
\newtheorem{proposition}[theorem]{Proposition}
\newtheorem{corollary}[theorem]{Corollary}
\newtheorem{lemma}[theorem]{Lemma}

\theoremstyle{remark}
\newtheorem{remark}[theorem]{Remark}
\theoremstyle{definition}
\newtheorem{definition}[theorem]{Definition}

\newtheorem{example}[theorem]{Example}

\newcommand{\beq}{\begin{equation}}
\newcommand{\eeq}{\end{equation}}
\usepackage{multirow}
\usepackage{amsfonts}
\title{Synchronous vs. asynchronous coalitions in multiplayer games, with applications
to guts poker}

\author{Kevin Buck}
\address{Indiana University, Bloomington, IN 47405}
\email{kevbuck@iu.edu}
\thanks{Research of K.B. was partially supported under NSF grants no. DMS-2154387 and DMS-2206105}
\author{Jessica Babyak}
\address{Indiana University, Bloomington, IN 47405}
\email{jt103@iu.edu}
\thanks{Research of J.B. was partially supported under NSF grants no. DMS-2154387 and DMS-2206105}
\author{Leah Dichter}
\address{Bowdoin College, Brunswick, ME 04011}
\email{ldichter@bowdoin.edu}
\thanks{Research of L.D. was partially supported under NSF grant no. DMS-2051032 (REU).}
\author{David Jiang}
\address{University of Wisconsin, Madison, WI 53706-1388}
\email{djiang38@wisc.edu}
\thanks{Research of D.J. was partially supported under NSF grant no. DMS-2051032 (REU).}
\author{Kevin Zumbrun}
\address{Indiana University, Bloomington, IN 47405}
\email{kzumbrun@iu.edu}
\thanks{Research of K.Z. was partially supported under NSF grants no. DMS-2154387 and DMS-2206105}

\begin{document}

\begin{abstract}
	We study the issue introduced by Buck-Lee-Platnick-Wheeler-Zumbrun of synchronous vs. asynchronous
	coalitions in multiplayer games, that is, the difference between coalitions with full and partial
	communication, with a specific interest in the context of continuous Guts poker where this
	problem was originally formulated.
	We observe for general symmetric multiplayer games, with players 2-n in coalition against player 1,
	that there are three values, corresponding to symmetric Nash equilibrium, optimal asynchronous,
	and optimal synchronous strategies, in that order, for which inequalities may for different examples
	be strict or nonstrict (i.e., equality) in any combination.
	Different from Nash equilibria and synchronous optima, which may be phrased as convex optimization problems,
	or classical 2-player games, determination of asynchronous optima is a nonconvex optimization problem.
	We discuss methods of numerical approximation of this optimum, and examine performance on 3-player rock-paper-scissors and discretized Guts poker.
	Finally, we present sufficient conditions guaranteeing different possibilities for behavior, based on
	concave/convexity properties of the payoff function.
	These answer in the affirmative the open problem posed by
	Buck-Lee-Platnick-Wheeler-Zumbrun whether the optimal asynchronous coalition value for 3-player guts
	is equal to the Nash equilibrium value zero.
	At the same time, we present a number of new results regarding synchronous coalition play
	for continuous $3$-player guts. 
\end{abstract}

\date{\today}
\maketitle
\tableofcontents

\section{Introduction}\label{s:intro}
In this paper, we examine the topic introduced in \cite{BLPWZ} 
of ``synchronous'' vs. ``asynchronous'' coalitions in multiplayer zero-sum games, with
a particular eye toward the context of Guts Poker studied there.
Consider an $n$-player symmetric zero-sum game, with payoffs $\Psi_i(s_1, \dots, s_n)$
to players $i$ given a choice of strategies $s_j\in \{ 1, \dots, N\}$ for $j=1,\dots, n$,
under the restrictions (i) (symmetry) $\Psi_i$ is invariant under transpositions of different $s_j$, $j\neq i$
and equal to $\Psi_j$ under transpositions of $s_i$ with $s_j$, and (ii) (zero sum) $\sum_{i=1}^n\Psi_i=0$.
Then \cite{N}, there is a symmetric {\it Nash equilibrium} consisting of identical mixed strategies,
or randomly chosen choices with probability densities $p_j=(p_{j1}, \dots, p_{jN}$, for which
departure by any single player results in a lower payoff for that player. The value of this Nash
equilibrium, by properties (i)-(ii), is necesssarily zero.

On the other hand, viewing this as a game between player 1 and a coalition of players 2-n working together, we may
ask what is the value of that (no longer symmetric) game, and what is its relation to the Nash equilibrium.
The {\it asynchronous optimum} is the optimal set of mixed strategies for players 2-n, which includes identical
mixed strategies, hence gives a value that is necessarily less than or equal to the Nash equilibrium value.
We call this asynchronous because the players 2-n of the coalition, though they may agree beforehand with which
probabilities each will choose strategies $\{1,\dots, N\}$, they are not allowed to coordinate these choices, but must
determine them independently.

The {\it synchronous optimum} on the other hand, is defined as an optimal probability distribution
\be\label{synch}
q_{j_2,\dots, j_n}, \qquad j_i\in \{1, \dots, N\},
\ee
among the possible $(n-1)N$ joint choices of players 2-n.
Evidently, these include the asynchronous strategies as the strict
subset
\be\label{asynch}
q_{j_2,\dots, j_n}=\Pi_{i=2}^n p_{i,j_i}
\ee
of strategies of tensorial form, hence the synchronous optimum gives a value less than or equal to that of
the asynchronous one.
We denote these values for convenience of discussion as 
\be\label{Vs}
V_S\leq V_A\leq V_N=0,
\ee
where $V_S$ is the synchronous coalition value, 
$V_A$ the asynchronous coalition value, and $V_N$ the symmetric Nash equilibrium value for player 1. 

We note, by the fundamental theorem of finite 2-player zero-sum games \cite{vN}, that
\ba\label{VS}
V_S &= \min_{q_{i_2,\dots,i_n}} \max_{i_1} \sum  q_{i_2,\dots, i_n} \psi(i_1, \dots, i_n)\\
&= \max_{p_{i_1}} \min_{i_2,\dots,i_n} \sum p_{i_1} \psi(i_1, \dots, i_n),
\ea
so that $V_S$ is also the maximum value forceable by player 1 by an optimal choice of mixed strategy,
whether against synchronous coalition, asynchronous coalition, or no coalition (independent play),
as this is computed by worst-case analysis as described in the final term of \eqref{VS}.
By comparison,
\ba\label{VA}
V_A &= \min_{p_{i_2},\dots,p_{i_n}}  \max_{i_1} \sum  \Pi_{j=2}^N p_{i_j} \psi(i_1, \dots, i_n).
\ea
Thus, {\it the asynchronous coalition game has a value}, in the classical sense \cite{vN} 
that the optimal values forceable by player 1 and by players 2-n agree, {\it if and only if $V_S=V_A$}.
Meanwhile, the synchronous coalition game, as a classical 2-player game, has value $V_S$.

On the other hand, {\it the asynchronous game is ``fair''}, in the sense that players 2-n do not gain an advantage
vs. player 1 by coalition, {\it if and only if $V_A=V_N=0$.}
Thus, the asynchronous game is both fair and has a well-determined value if and only if
\be\label{cdet}
V_S=V_A=V_N=0,
\ee
i.e., there is no advantage gained by coalition of any sort between players 2-n.
This completely determinate case is somewhat rare among multiplayer games, however, as typically \cite{O,BLPWZ}, 
\be\label{typical}
V_S<V_N=0,
\ee
a major practical/philosophical complication in the analysis of multiplayer games.  

The above concepts have natural generalizations to continuous symmetric zero-sum games 
like continuous Guts Poker \cite{CCZ,BLPWZ},
with the arguments $i_j$ of $\Psi$ replaced by real numbers $x_j$, and discrete probability distributions 
$p_j$ replaced by continuous ones.

\subsection{Objectives}\label{s:objectives}
The discrepancy \eqref{typical} between synchronous coalition and Nash equilibrium values
is well-documented in the literature, and has been substantially discussed.
Our purpose here, building on observations of \cite{BLPWZ}, is to examine under what circumstances
there is a discrepancy 
\be\label{Atyp}
V_A<V_N=0
\ee
between asynchronous coalition and Nash equilibrium values,
a question that seems to have been not at all considered, but one that seems relevant and important
to dynamics of multiplayer games.
For, $V_A=V_N=0$ would imply that, in the absence of communication between players 2-n, as one could
imagine in various real-world scenarios, there is no advantage forceable by a player 2-n coalition,
yet, at the same time, player 1 cannot force a zero return against all possible plays.

In particular, we consider the key open problem posed in \cite{BLPWZ} whether or not
$V_A=V_N=0$ for continuous Guts Poker.
More generally, we investigate techniques for numerical approximation of $V_A$ for arbitrary games,
an interesting nonconvex optimization problem.

\subsection{Results}\label{s:results}
Our main analytical results are, first, to characterize all possible behaviors of synchronous vs. asynchronous
coalition values for $2\times 2\times 2$ symmetric zero sum games, second, for a simple $3\times 3\times 3$
version of Rock-Paper-Scissors, to show that there can occur the new possibility of a gap between the
asynchronous value $V_A$ and both the Nash equilibrium value $V_N=0$ and the synchronous value $V_S$.
Third, we treat in detail the example of continuous 3-player guts introduced in \cite{CCZ}, obtaining
an explicit analytical solution also of this at first sight quite complicated infinite-dimensional case.
Here, the main idea is, by using partial convexity properties, to reduce to consideration of pure strategy
solutions, converting the maximin and minimax problems defining $V_S$ and $V_A$ to finite-dimensional
calculus problems \eqref{VSmaximin} and \eqref{VAminimax} analogous to \eqref{VS} and \eqref{VA},
which may with some effort be solved completely.

The latter appears of independent interest,
{\it identifying continuous guts as a rare instance of a realistic 3-player game admitting complete solution.}
Specifically, we (i) answer in the affirmative, by rigorous analysis, 
the key open problem posed in \cite{BLPWZ} whether or not $V_A=V_N=0$ for continuous Guts Poker,
and (ii) confirm rigorously the optimal synchronous coalition strategies and values observed
numerically in \cite{BLPWZ}. 
Moreover, using our solution formulae, we are able to track the minute evolution of optimal strategies
for the recursive game as the number of rounds increases, a delicate computation not easily accessible
by numerical approximation. 

In the remainder of the paper, we investigate numerical methods for approximating synchronous and asynchronous 
values for general games, an interesting class of convex ($V_S$) and nonconvex ($V_A$) optimization problems,
using the exact solutions from (i) and (ii) above as benchmark problems, then investigating
convergence/computational cost on testbeds of larger randomly generated 2- and 3-player symmetric games.
Finally, having tested our methods,
we numerically investigate frequency distributions of the gaps between $V_S$, $V_A$, $V_N$ for random
symmetric $3$-player games.

\subsection{Plan of the paper}\label{s:plan}
In Sections \ref{s:eg1}-\ref{s:eg3}, we present some basic discrete examples illustrating the
possible behaviors for low-dimensional games, and deriving values and optimal strategies for
synchronous vs. asynchronous games. In Section \ref{s:eg4}, we recall the description
of continuous Guts Poker in \cite{CCZ}, and discretizations thereof, and rigorously determine the
value of the asynchronous game to be zero, answering the open problem posed in \cite{BLPWZ}.
At the same time, we show how to recover by rigorous analysis all features relevant to the synchronous
coalition game, validating numerical observations of \cite{BLPWZ} and giving useful benchmarks for
our numerical studies to follow.
In Section \ref{s:num}, we discuss methods for numerical approximation, using our exact
solutions as useful guidelines for comparison.
In Section \ref{s:numeff}, we compare accuracy/computational effiency of these various methods by
experiments on randomly chosen two-player games of different sizes.
In Section \ref{s:numV}, we use them to compile statistics for randomly generated three-player games.
on gaps $|V_S-V_N|$, $|V_S-V_A|$, $|V_A-V_N$, and relative gap $\theta:=|V_A-V_N|/|V_S-V_N|$. 
Finally, in Section \ref{s:disc}, we present discussion and open problems.
In Appendix \ref{s:critical}, we provide for completeness a discussion of maximin vs. critical points.
In Appendix \ref{s:alphabeta}, we study evolution of optimal strategies for recursive games with the
round of play; see in particular Section \ref{s:rguts} on continuous guts.
In Appendix \ref{s:Tables1}, we collect tables of results having to do with Section \ref{s:numeff}.

\medskip
{\bf Acknowledgement:} This work was carried out with the aid of open source packages
Desmos, Nashpy, and SciPy. L.D. and D.J. thank Indiana University, especially
REU director Dylan Thurston and administrative coordinator Mandie McCarty, for their
hospitality during the REU program in which this work was carried out.
We also thank the UITS system at Indiana University for the use of supercomputer cluster Carbonate.
Code used in the investigations of this project is publicly available and may be found at \cite{Gi}.

\section{Example 1: Odds and Evens}\label{s:eg1}
Simple illustrative examples are given by versions of the game of odds and evens.
In the most basic version, each of three player chooses ``one'' or ``two'' and the players simultaneously
display their choices. In the ``odd man out'' (OMO) version, if two players match and the other does not,
then the ``odd'' (nonmatching) player pays one value unit to each of the other two players. If all match, then
the game is a tie, with return to all players of zero.
In the ``odd man in'' (OMI) version, if two players match and the other does not, then the ``even'' (matching)
players each pay one unit to the odd player. If all match, as in (OMO), there is a tie and the value is zero.
Both versions are clearly symmetric, so have symmetric Nash equilibria returning value zero.

\subsection{Values for (OMO)} 
Evidently, the unique symmetric Nash equilibrium strategy is for each player to choose
one or two with equal probability $1/2$, guaranteeing a return of $V_N=0$.
On the other hand, players 2-3 can choose strategy pairs $(1,1)$ and $(2,2)$ with equal probability
to guarantee a return of $(-2)(1/2) + (0)(1/2)=-1$ to player 1.
Meanwhile, player 1 can choose $1$ and $2$ with equal probability to guarantee a return of $\geq -1$ against
synchronized coalition strategies.  Thus, $V_S=-1$.

Finally, setting $0\leq y\leq 1$ to be the probability that player 2 chooses ``one'', and $0\leq z\leq 1$ the
probablity that player 3 chooses ``one'', we find that the expected return to player 1 upon choosing ``one''
is
\be\label{OMO1}
(0)yz + (1)y(1-z) + (1) (1-y)z + (-2)(1-y)(1-z)=
3y+3z-4yz-2
\ee
while the expected return on choosing ``two'' is
\be\label{OMO2}
(-2)yz + (1)y(1-z) + (1) (1-y)z + (0)(1-y)(1-z)= y+z -4yz.
\ee
A quick computation shows that the second majorizes the first precisely when $y+z\leq 1$.
Minimizing \eqref{OMO2} on $0\leq y,z$, $y+z\leq 1$, we find a saddle at $(y,z)=(1/4,1/4)$,
giving value $1/2- 1/4=+1/4$. On the boundaries $y=0$ or $z=0$, the return is $z$ or $y$, respectively,
both $\geq 0$. By symmetry, we have the same behavior on $y=1$ or $z=1$.
On the boundary $y=1-z$, the return is $1-4z(1-z)= (1-2z)^2\geq 0$ as well.
By symmetry, the best that players 2-3 can force for $y+z\geq 1$ is $0$, too, showing that $V_A=0$.

The above analysis shows in passing that the (nonconvex) asynchronous problem has no local minimizers
other than the global minimizers recorded in Proposition \ref{omoprop}.
The interior minimizer $(y,z)=(1/2,1/2)$,  
%found on the boundary between the first and second regions described,
corresponding to a minimax at $(x,y,z)=(1/2,1/2,1/2)$ of the expected return 
\be\label{r1}
\alpha(x,y,z)=-4yz+2x(y+z-1)+ (y+z)
\ee
for $0\leq x\leq 1$ defined as the probability that player one chooses ``one'',
is a critical point of $\alpha$, in agreement with Proposition \ref{critprop}, below;
indeed, it is the unique critical point.

Summarizing, we have as follows.

\begin{proposition}\label{omoprop}
	For 3-player odds and evens (OMO),
\be\label{omoval}
V_S=-1< V_A=V_N=0.
\ee
	Moreover, the asynchronous minimax problem has global minimimizers $(y,z)=(0,0), (1,1), (1/2,1/2)$,
	with no other local minima. Local saddle points however occur.
\end{proposition}

\subsection{Values for (OMI)}\label{s:omi}
Again, the unique symmetric Nash equilibrium strategy is for each player to choose
one or two with equal probability $1/2$, guaranteeing a return of $V_N=0$.
But, now, , players 2-3 can guarantee a return of $-1$ to player 1 by choosing pure strategy
pairs $(1,2)$ and $(2,1)$ with arbitrary probabilities, including pure strategy pairs $(1,2)$ or $(2,1)$.
As pure strategies are also asynchronous strategies, this gives $V_A=V_S=-1$.

Observing that the expected payoff analogous to \eqref{r1} is now its negative,
$\alpha(x,y,z)=-[-4yz+2x(y+z-1)+ (y+z)]$, we find that there are no local minima for the 
asynchronous problem other than the global minimizers $(y,z)=(0,1)$ and
$(y,z)=(1,0)$ found on the boundary $y+z=1$, where $\alpha$ reduces to $- (1-2z)^2\leq 0$.
Note that these minimaxes are not critical points of $\alpha$, a fact associated with their
nonuniqueness; see Remark \ref{ceg}.

Summarizing, we have as follows.

\begin{proposition}\label{omoprop2}
	For 3-player odds and evens (OMI),
\be\label{omival}
V_S=V_A=-1< V_N=0.
\ee
Moreover, the asynchronous minimax problem has global minimizers $(y,z)=(0,1), (1,0)$,
with no other local minima.
\end{proposition}

\section{Example 2: general $2\times 2\times 2$ games}\label{s:eg2}
More generally, payoffs for a general symmetric zero-sum $2\times 2\times 2$ game,
by symmetry, take the form
\be\label{gensym}
\bp P_{111} & P_{112}& P_{121} & P_{122}\\
 P_{211} & P_{212}& P_{221} & P_{22,2}\ep =
 \bp 0 & \alpha& \alpha & -2\beta\\ -2\alpha  & \beta & \beta  & 0 \ep,
\ee
where $P_{i_1i_2 i_3}$ denotes the payoff to player 1 if each player $j$ chooses $i_j$.

Case $\alpha=\beta=0$ corresponds to the trivial zero game, with $V_S=V_A=V_N=0$.
If $\alpha$ or $\beta$ is nonzero, taking without loss of generality $\beta\neq 0$, by
symmetry, we may rescale by $|\beta|$ to obtain either
\be\label{redgensym}
\bp P_{111} & P_{112}& P_{121} & P_{122}\\
 P_{211} & P_{212}& P_{221} & P_{22,2}\ep =
 \bp 0 & \alpha& \alpha & -2 \\ -2\alpha & 1 & 1  & 0 \ep
\ee
or
\be\label{2redgensym}
\bp P_{111} & P_{112}& P_{121} & P_{122}\\
 P_{211} & P_{212}& P_{221} & P_{22,2}\ep =
 \bp 0 & -\alpha& -\alpha & 2 \\ 2\alpha & -1 & -1  & 0 \ep.
\ee
Case $\alpha=\beta=1$ corresponds to (OMO), while $\alpha=\beta=-1$ corresponds to (OMI).
Thus, \eqref{redgensym} may be considered as a weighted payoff (OMO), where the ``odd man'' penalty
depends on the strategy chosen by the odd man, and \eqref{2redgensym}, similarly, as a weighted version
of (OMI).

\subsection{Case $\alpha\geq 0$, \eqref{redgensym}: generalized (OMO)}
Evidently, columns 2 and 3 of the righthand side of \eqref{redgensym} are inferior strategies for
player 2, and can be ignored in the synchronous coalition game.
This reduces the problem to a $2\times 2$ two-player game with payoff matrix

\be\label{redbloc}
 \bp 0 &  -2 \\ -2\alpha &  0 \ep,
 \ee
 for which a quick computation gives 
 %$V_S= - \frac{4\alpha}{2\alpha+2}$, with optimal player 1 strategy
 %choosing 1 with probability $ \frac{2\alpha}{2\alpha+2}$ and 2 with probability $ \frac{2}{2\alpha+2}$.
 $V_S= - \frac{2\alpha}{\alpha+1}$, with optimal player 1 strategy
 choosing 1 with probability $ \frac{\alpha}{\alpha+1}$ and 2 with probability $ \frac{1}{\alpha+1}$.

 %$V_A$ now.
 Turning to the 
 computation of $V_A$, set $y$ and $z$ to be the probabilities that player 2 and player 3 choose $1$. 
 Then, the associated payoffs may be computed as
 \be\label{gOMOpay1}
 -2 -2(1+\alpha)yz + (2+\alpha)(y+z)
 \ee
 when player 1 chooses $1$, and
 \be\label{gOMOpay2}
  -2(1+\alpha)yz + y+z
 \ee
 when player 1 chooses $2$. The second majorizes the first when 
 \be\label{gOMOmaj}
 y+z\leq \frac{2}{1+\alpha},
 \ee
 the first majorizing the second on the complement.
 Minimizing \eqref{gOMOpay2} over \eqref{gOMOmaj}, we find a single critical saddle point
 at $(y,z)=(\frac{1}{2(1+\alpha)} ,\frac{1}{2(1+\alpha)})$, which can therefore not be an interior minimum.
 %NOTE: value
 $\frac{1}{2(1+\alpha)}$. 
 On the boundaries $y=0$ and $z=0$, we have payoffs $z\geq 0$ and $y\geq 0$,
 respectively, returning at least zero.
 On the boundary $y+z =\frac{2}{1+\alpha}$, the payoff may be calculated to be 
 $$
 -2(1+\alpha)y\Big( \frac{2}{1+\alpha}-y\Big) + \frac{2}{1+\alpha}= 
 2\Big(\sqrt{1+\alpha} y- \frac{1}{\sqrt{1+\alpha}}\Big)^2,
 $$
 hence greater than or equal to zero.

 When $\alpha<1$, $ \frac{2}{1+\alpha}>1$, and there are two other boundaries $y=1$ and $z=1$, without loss
 of generality (by symmetry) $y=1$, $0\leq z\leq  \frac{1-\alpha}{\alpha+1}$.
 On this boundary, we have payoff
 $$
 -2(1+\alpha)z + 1+z= 1 -(1+2\alpha)z,
 $$
 which is minimized at $z =\frac{1-\alpha}{\alpha+1}$, with value
 $$
 P_*= 1- (1+2\alpha) \frac{1-\alpha}{1+\alpha}=
 %\frac{ 1+\alpha - (1+2\alpha(1-\alpha)   } {1+\alpha}=
 \frac{ 2\alpha^2   } {1+\alpha}\geq 0.
 $$
 
 A similar computation on the complement of \eqref{gOMOmaj} gives the
 same result, hence $V_A\equiv 0$ for all $\alpha\geq 0$ for this class of games.
	 (Indeed, the invariance $y\to 1-y$, $z\to 1-z$, $\alpha \to 1/\alpha$ reduces this to the
	 previously considered case.)

\subsection{Case $\alpha\geq 0$, \eqref{2redgensym}: generalized (OMI)}
Here, columns 1, 3, and 4 of the righthand side of \eqref{2redgensym} are majorized for player 2 by
column 3, so may be ignored in the synchronous coalition game. This gives a trivial $2\times 1$ reduced game
with payoff matrix
\be\label{redbloc2}
 \bp -\alpha   \\  -1   \ep,
 \ee
 evidently returning value $V_S=\max\{-1,-\alpha\}$ to player 1.
 As this corresponds to a pure, or deterministic strategy pair $(1,2)$ for players 2-3,
 we have in this case $V_S=V_A<V_N=0$, similarly as in the basic (OMI) case of Section \ref{s:omi}.

\subsection{Case $\alpha \leq 0$} 
In the case $\alpha\leq 0$, \eqref{redgensym}, player 1 can force a return of zero by
the choice $2$, hence $V_S=V_A=V_N=0$.
In case $\alpha\leq 0$, \eqref{2redgensym}, player 1 can force a return of zero by
the choice $1$, hence again $V_S=V_A=V_N=0$.

\subsection{Summary}\label{s:sum}
Collecting the above conclusions, we have the following results categorizing possible forms and
behavior for general symmetric $2 \times 2\times 2$ zero-sum games.

\bt\label{genthm}
Any symmetric $2\times 2\times 2$ zero sum game may be reduced by rescaling/symmetry to 
either the trivial zero game, or a game of form \eqref{redgensym} or \eqref{2redgensym}. 
In the first case, $V_S=V_A=V_N=0$, and in the second $V_S<V_A=V_N=0$.
In the third case, $V_S=V_A<V_N=0$ for $\alpha>0$ and $V_S<V_A=V_N=0$ for $\alpha\leq 0$.
\et

\section{Example 3: three-player Rock-paper-scissors}\label{s:eg3}
We next examine two related $3\times 3\times 3$ odds and evens games discussed in \cite[Appendix D]{BLPWZ},
which could be considered as three-player generalizations of Rock-Paper-Scissors.
These illustrate that moving from $2\times 2\times 2$ to $3\times 3\times 3$ games opens
up the final new possibility for behavior of
$$
V_S<V_A<V_N=0.
$$
The descriptions and analysis below are paraphrased from \cite[Appendix D]{BLPWZ},
except for the discussion of local minimizers which is new.

\subsection{Odd man in}\label{s:in}
In this version, each player chooses a value 1, 2, or 3 for "Rock", "Paper", or "Scissors". 
If all choices are the same, or all are different, there is no payoff. 
If two players choose a common number, however, and the third player
a different one, then the first two each pay a value of 1 to the third, i.e., the first two receive payoff -1
and the third +2. Clearly, the strategy distribution $(1/3, 1/3, 1/3)$ for player 1 gives average 
return of $+2/3$ if the other two players play the same number, and $-2/3$ if they play different numbers.
Thus, player 1 can force $\geq -2/3$.  On the other hand,  if players 2-3 choose with equal probability
between pairs of choices $(1,2)$, $(1,3)$, and $(2,3)$, then the average payoff to player 1 is independent
of player 1's choice of strategy, and equal to $(2/3)\times (-1) + (1/3)\times (0)=-2/3$.  Thus, players
2-3 can force a return of $\leq -2/3$ to player 1 by synchronous coalition play, and the value of the 
player 1 vs. players 2-3 game is $-2/3$. 

On the other hand, let $y:=(y_1,y_2,y_2)$ and $z:=(z_1,z_2,z_3)$
denote probability distributions describing mixed strategies for players 2 and 3.
Then, it is readily computed \cite{BLPWZ} that the payoff to player 1 is
\ba\label{aspay}
&\hbox{ \rm $\Psi_1(y,z):=2y\cdot z - (y_1+z_1)$ for player 1 choice 1,}\\
&\hbox{ \rm $\Psi_2(y,z):=2y\cdot z - (y_2+z_2)$ for player 1 choice 2,}\\
&\hbox{ \rm $\Psi_3(y,z):=2y\cdot z - (y_3+z_3)$ for player 1 choice 3.}\\
\ea
The value $\Psi(y,z):= \max_j \Psi_j(y,z)$ is thus the minimum value forceable by choice $(y,z)$,
and
$$
\overline{V}=\min_{y,z} \Psi(y,z)
$$
is the minimum value forceable by players 2-3 via asynchronous play, and by continuity of $\Psi$ is
achieved for some feasible pair of strategies $(y_*,z_*)$.

Noting that the average of $\Psi_j$ is 
$$
2 y\cdot z - (1/3)\sum_j (y_j+z_j)= 2 y\cdot z - (2/3)\geq -2/3,
$$
we have that $\Psi(y,z)\geq -2/3$, with equality if and only if simultaneously $y\cdot z=0$
and $(y_j+z_j)=2/3$ for all $j$. But, these together imply that one of each pair $y_j, z_j$ has
value zero and the other value $2/3$, which is impossible to reconcile with $\sum_j y_j=\sum_j z_j=1$.
Thus, evaluating at $(y,z)=(y_*,z_*)$, we obtain $\overline{V}= \Psi(y_*,z_*)>-2/3$, verifying
that there is indeed a gap between this value and the value $-2/3$ forceable by synchronous coalition play.
Indeed, the optimum asynchronous strategy 
can be shown to be $y_*=(1,0,0)$, $z_*=(0, 1/2, 1/2)$, and symmetric permutations thereof, forcing an expected
payoff to player 1 of $\leq -1/2$: thus, 
$$
-2/3=V_S< V_A=-1/2< V_N=0,
$$
leaving a gap of $-1/2-(-2/3)=1/6$ between $V_S$ and $V_A$ and a gap of $1/2$ between $V_A$ and $V_N$.

As regards local minima for the asynchronous game, we may first check readily that the unique critical points
in each of the three regions above are $y=(4/6,1/6,1/6)$, $z= (4/6,1/6,1/6)$ for $\Psi_1$, and symmetric
rearrangements for $\Psi_2$ and $\Psi_3$. But, these are outside the ranges of validity of the $\Psi_j$;
for instance, $\Psi_1$ is valid only where $y_1+z_1$ minimizes $y_j+z_j$.
Let us check next on the boundary $y_1+z_1=y_2+z_2=h$, and symmetric rearrangements. Here, we find that
the unique critical points are $y= z= (5/12,5/12, 2/12)$, and symmetric rearrangements.  Again, this
is out of the range of validity.  

Checking on the triple interior boundary $y_1+z_1=y_2+z_2= y_3+z_3=2/3$, 
we may set $z_1=2/3-y_1$, $z_2=2/3-y_2$, and minimize the resulting function
\be\label{triplered}
\check \psi(y_1,y_2):= 2\Big( y_1(2/3-y_1)+ y_2(2/3-y_2)\Big) - 2/3
\ee
on the domain
\be\label{dom}
1/3\leq y_1+y_2 \leq 1, \qquad 0\leq y_1,y_2\leq 2/3.
\ee
We find that the unique interior critical point is the Nash equilibrium $y=z=(1/3,1/3,1/3)$, 
which is (automatically) in the range of validity. However, this is a maximum and not a minimum for
the problem \eqref{triplered}, since $2 w(2/3-w)$ is maximized at $w=1/3$.
The minima with respect to this restricted problem thus occur at the boundary points 
$$
(y_1,y_2)=(0,1/3), (1/3,0), (2/3,0) (2/3,1/3), (1/3,2/3), (0, 2/3),
$$
giving values $\Psi(y,z)=-4/9>-1/2=V_A$. Thus, they are at best local and not global minimizers,
and are the only candidates for local minimizers lying on the triple interior boundary.

Further analysis yields that they are indeed local minimizers with respect to general admissible perturbations
as well. 
For, taking without loss of generality $y=(0,1/3,2/3)$, $z=(2/3, 1/3,0)$, and assuming by symmetry that
the perturbed $y$, $z$ feature either a) $y_1+z_1$ is minimum, or b) $y_2+z_2$ is minimum, let us consider
each case in turn. In case a), we must have
$$
y=(\theta,1/3+ \gamma + \beta + \delta -\theta, 2/3-\gamma -\beta-\delta), 
$$
$$
z=(2/3-\theta-\beta, 1/3+\theta + \beta -\gamma, \gamma),
$$
with $\theta,\gamma,\delta\geq 0$ and $\beta>0$, giving 
$$
\psi(y,z)=
2\Big(  \theta(2/3-\theta -\beta) + (1/3 +\gamma + \beta + \delta -\theta)(1/3+ \theta +\beta -\gamma)
+ (2/3-\gamma -\beta-\delta)\gamma \Big) - (2/3-\beta),
$$
giving first variation
$$
(2/3)(2\theta + \gamma + \beta + \delta -\theta + \theta + \beta -\gamma + 2\gamma) - \beta=
%(2/3)(2\theta + 2\beta + \delta +  2\gamma) - \beta=
(2/3)(2\theta  + \delta +  2\gamma) + \beta/3>0.
$$

Similarly, in case b), we must have
$$
y=(\gamma,1/3+ a ,2/3- a - \gamma ),
$$
$$
z=(2/3-\gamma -\theta+\delta, 1/3-a- \theta, \gamma+ 2 \theta -\delta +a), 
$$
with $\beta,\gamma,\delta\geq 0$, 
$\gamma+ 2 \theta -\delta +a\geq 0$, and $\theta>0$, giving first variation
$$
(2/3)\Big(2\gamma +a -a -\theta + 2(\gamma +2\theta -\delta + a)\Big) +\theta\geq \theta/3>0.
$$
Taken together, this verifies that $y=(0,1/3 ,2/3 )$, $z=(2/3, 1/2,0)$ is a (nonsmooth) local minimizer,
and similarly for its symmetric rearrangements.

Finally, a tedious case-by-case analysis, omitted, shows that there are no local minimizers on the boundary of
the domain $y_1,y_2,z_1,z_2\geq 0$, $y_1+y_2, z_1+z_2 \leq 1$, other than the {\it global minimizers}
$(y_1,y_2)=(1,0)$, $z=(0,1/2)$; 
$(y_1,y_2)=(0,1)$, $z=(1/2,0)$; 
and $(y_1,y_2)=(0,0)$, $z=(1/2,1/2)$,
together with the local minimizers just determined.  This accounts for all global and local minimizers.
Finally, returning to the Nash equilibrium, we note that by definition it is minimum with respect to
perturbations involving one player at a time, so is neither a local minimum nor a local maximum, but
a nonsmooth saddle.

We record this as follows.

\begin{proposition}\label{RPSOMIprop}
	For 3-player Rock-Paper-Scissors (OMI),
\be\label{RPSomival}
-2/3=V_S< V_A=-1/2< V_N=0.
\ee
Asynchronous global minima are achieved at $y=(1,0,0),z=(0,1/2,1/2)$;
$y=(0,1,0),z=(1/2,0,1/2)$; and $y=(0,0,1),z=(1/2,1/2,0)$, while asynchronous local minima are
achieved at 
$y=(0,1/3,2/3),z=(2/3,1/3,0)$; 
$y=(2/3,1/3,0),z=(0,1/3,2/3)$; 
$y=(1/3,0, 2/3),z=(1/3,2/3,0)$; 
$y=(1/3,2/3, 0),z=(1/3,0, 2/3)$; 
$y=(0,2/3, 1/3),z=(2/3, 0, 1/3)$; 
and $y=(0,1/3,2/3),z=(2/3, 1/3,0)$. 
The Nash equilibrium $y=z=(1/3,1/3,1/3)$, meanwhile, is a nonsmooth saddle.
\end{proposition}

\br\label{locrmk}
The appearance of local (nonglobal) asynchronous minimizers in Proposition \ref{RPSomival} 
is significant as a major difference from the convex synchronous minimization problem,
present already in this simple setting.
Certainly it complicates numerical estimation of minima discussed in Section \ref{s:num} by
descent or other iterative methods, as randomly chosen starting data may lie in the basin of
attraction of a local but not global minimizer, thus returning a local and not the correct global minimum
as numerical approximation.
\er

\br\label{KKTrmk}
In more complicated situations, the  search for nonsmooth minimizers would be more systematically done 
by checking Karush-Kuhn-Tucker conditions \cite{BMS}.
\er

\subsection{Odd man out}\label{s:out}
Next, consider the same game, but with payoff function multiplied by $-1$: that is, the ``reverse'' game,
in which the odd player is penalized instead of rewarded.
Here, an optimum synchronized strategy is a blend of pure strategy pairs 1-1, 2-2, 3-3, 
each chosen with probability $1/3$, yielding value $(2/3)(-2)=-4/3$.
The symmetric Nash equilibrium may be computed \cite{BLPWZ} to be
$x=y=z=(1/3,1/3,1/3)$, returning payoff zero to all players.  

However, the optimum value forceable by asynchronous coalition of players 2-3 is now
$
\tilde{V}= \min_{y,z}\tilde \Psi(y,z)$, where
$$
\tilde \Psi(y,z): =\min -\psi_j(y,z)= -2 y\cdot z + \max_j (y_j+z_j),
$$

The effect of changing the sign of the payoff function in going from OMI to OMO, is to take $\Psi_j$ to
$-\Psi_j$, but changing the domains of validity from $(y_j+z_j)$ minimal (OMI) to 
to $(y_j+z_j)$ maximal (OMO).
This has the effect that the disallowed critical points
$y=(4/6,1/6,1/6)$, $z= (4/6,1/6,1/6)$ for $\Psi_1$, and symmetric versions for $\Psi_2$ and $\Psi_3$, are
now allowed, so that valid interior critical points do occur. On the other hand, the associated Hessian
$-2y\cdot z$ is indefinite, hence they are not local minimizers, but saddles.

Similarly, on the interior boundary $y_1+z_1=y_2+z_2=:h$, we find that the formerly disallowed critical point
$y=z=(5/12,5/12)$ (and symmetric permutations thereof) is now allowed, hence must be checked further.
However, considering the restricted class of competitors $y_1=y_2=z_1=z_2$ gives payoff function
$\psi(y_2)=-12y_2^2+10y_2 -2$, which satisfies $\psi'(y_2)= -24y_2 + 10$ vanishing at $y_2=5/12$,
but $\psi''(y_2)=-24<0$ showing that it is a maximum and not a minimum. Hence, this too can be discarded,
corresponding to a nonsmooth saddle.
The endpoints $h=0$ and $h=1$, however, give $(y_1,y_2)=(z_1,z_2)=(0,0)$ and
$(y_1)=(z_1)=(0)$, respectively, reducing the problem to 
a trivial one-strategy and the two-strategy Odds-Evens case, each of which feature $V_A=0$.
Hence, these furnish in the first case the minimizer $y=z=(0,0,1)$ and in the second
the previously determined minimizers $(y,z)=(1,0,0)$,
$(y,z)=(0,1,0)$, and $y=z=(1/2,1/2,0)$ of the two-strategy case, along with symmetric rearrangements.

Finally, the interior double boundary $y_1+z_1=y_2+z_2=y_3+z_3$ 
yields as before the Nash equilibrium $y=z=(1/3,1/3,1/3)$, or global minimimum, as the unique interior 
minimizer.
On the domain boundaries $y_1=0$, $y_2=0$, $y_1+y_2=1$, $z_1=0$, $z_1=1$, $z_1+z_2=1$,
meanwhile, there appear the already noted global minimizers $y=z=(1,0,0)$, and symmetric rearrangements 
thereof, but no other local minimizers. All of the above-described minima have value zero.

Thus, in the reverse (OMO) game,
$$
-4/3=V_S < V_A=V_N=0.
$$
We record this as follows.\footnote{This repairs an error in \cite{BLPWZ}, both in the proof, and
the resulting omission of several minimizers.}

\begin{proposition}\label{RPSOMOprop}
For 3-player Rock-Paper-Scissors (OMO),
\be\label{RPSomoval}
-4/3=V_S < V_A=V_N=0.
\ee
Asynchronous global minima are achieved at $y=(1/3,1/3,1/3),z=(1/3,1/3,1/3)$;
	$y=(1,0,0),z=(1,0,0)$; $y=(0,1,0),z=(0,1,0)$; and $y=(0,0,1),z=(0,0,1)$;
	and 
	$y=(1/2,1/2,0),z=(1/2,1/2,0)$, $y=(0,1/2,1/2),z=(0,1/2,1/2)$, and $y=(1/2,0 , 1/2),z=(1/2,0, 1/2)$,
	with no other local minima.
\end{proposition}

\br\label{maximinrmk}
In all of these examples, a more straightforward way to compute the synchronous value $V_s$ is, using the
fundamental theorem of games, to evaluate the maximin 
$$
\hbox{\rm $\max_x  \sum_i\min_{jk}x_i P_{ijk}=:\max_x \Phi(x)$, where $\Phi(x):= \sum_i \min_{jk}x_i P_{ijk}$.}
$$
For example, in OMI, it is readily seen that $\Phi(x)= \min_j (x_j) -1$, so that $\min_x \Phi=-2/3$, achieved
at $x=(1/3,1/3,1/3)$
Similarly, in OMO, it is found that $\Phi(x)= 2( \min_j (x_j) -1)$, so that $\min_x \Phi=-4/3$, achieved
again at $x=(1/3,1/3,1/3)$.
\er

\section{Example 4: continuous Guts Poker and discretizations}\label{s:eg4}
Finally, we come to our main example, the continuous version of Guts Poker introduced in \cite{CCZ},
which can be played with any number of players $n \geq 2$.
In this game, players make an initial one unit ante into a pot, and are dealt continuous ``hands'' consisting
of I.I.D. random variable uniformly distributed on $[0,1]$.
On the count of three players either ``hold'' or ``drop'' their hands, with no further betting or cards dealt.
If only one player holds, they win the pot and the round is terminated.  If no players hold, the game is
redealt, starting over.  If $m\geq 2$ players hold, the player with highest ``hand'' wins the pot and the 
remaining $m-1$ players must match it, so that the stakes increase by factor $m-1$. A new hand
is then dealt to all players and the game played in the same way but with now higher stakes, this process
continuing until play is terminated by a single player holding.

As described in \cite{CCZ,BLPWZ}, the study of this variable-stakes ``generalized recursive game''\footnote{
	See \cite{E,Sh3} for related notions of recursive and stochastic games.}
can be reduced to the study of the ``single-shot'' game consisting of the outome of a single round.
A ``pure'' strategy for player $i$ for the one-shot game, indexed by $p_i^*\in [0,1]$, is the threshold type strategy
to hold for $p_i> p_i^*$ and otherwise drop.
A ``mixed,'' or ``blended'' strategy is a random mixture of pure strategies with a given probability weight.
The outcome for the single-shot game may be encoded by expected instantaneous return $\alpha(p_1,\dots, p_n)$ in
that round for a selection of pure strategies, together with expected stakes $\beta(p_1,\dots, p_n)$ 
for ensuing rounds.

It was shown in \cite{CCZ,BLPWZ} that players 2-$n$ working in synchronous coalition may force a negative
return for player 1 for the full recursive game if and only if they may force a negative instantaneous return for
the one-shot game with payoff function $\alpha(\dots)$.
Hence, we may focus in this discussion on the one-shot game, a classical continuous $n$-player game with no
recursive aspect.
It was shown analytically and numerically in \cite{CCZ,BLPWZ} that $V_S<0$ for this game.
Here, we investigate the remaining open question posed there whether the corresponding asynchronous
coalition value $V_A$ is equal to $V_S<0$ or $V_N=0$, or lies in the open interval $(V_S, 0)$,
showing for the three-player version that in vact $V_A=V_N=0$.

%%%%%%%%%%%%%%%%%%%%%%%%%%%%%%%%%%%%%%%%%%%%%%%%%%%%%%%%%%
\subsection{Three-player payoff function for one-shot continuous Guts}\label{s:gutsfrag}
We start by recalling from \cite{CCZ}, without repeating the derivation,
the payoff function for one-shot continuous Guts.

\begin{proposition}[\cite{CCZ}]\label{3prop}
The one-shot payoff function for $3$-player continuous Guts is
	\ba\label{3alpha}	
	\alpha(p_1^*,p_2^*,p_3^*)&=
		\begin{cases}
		2p_1^*-p_2^*-p_3^*+(p_3^*)^3+3(p_2^*)^2p_3^*-4p_1^*p_2^*p_3^*, &
		p_1^*<p_2^*<p_3^*,\\
		2p_1^*-p_3^*-p_2^*+(p_2^*)^3+3(p_3^*)^2p_2^*-4p_1^*p_2^*p_3^*, &
		p_1^*<p_3^*<p_2^*,\\
		2p_1^*-p_2^*-p_3^*+(p_3^*)^3-3(p_1^*)^2p_3^*+2p_1^*p_2^*p_3^*, &
		p_2^*<p_1^*<p_3^*\\
		2p_1^*-p_2^*-p_3^*+(p_2^*)^3-3(p_1^*)^2p_2^*+2p_1^*p_2^*p_3^*, &
		p_3^*<p_1^*<p_2^*,\\
		2p_1^*-p_2^*-p_3^*-2(p_1^*)^3+2p_1^*p_2^*p_3^*, &
		p_2^*<p_3^*<p_1^*,\\
		2p_1^*-p_2^*-p_3^*-2(p_1^*)^3+2p_1^*p_2^*p_3^*, &
		p_3^*<p_2^*<p_1^*.
	\end{cases}\\
	%\beta&= 2-p_1^*-p_2^*-p_3^*+2p_1^*p_2^*p_3^* .
	\ea
\end{proposition}

We next derive some basic properties we will need in the analysis.

\bl\label{C1lem}
The payoff function $\alpha$ of \eqref{3alpha} is $C^1$ on its entire domain $p_*\in [0,1]^3$.
It is concave in $p_1^*$ and individually convex in $p_2^*$ and $p_3^*$, but not jointly
convex in $(p_2^*,p_3^*)$.
\el

\begin{proof}
By symmetry, we may take without loss of generality $p_2^*\leq p_3^*$.
Computing partial derivatives in the resulting cases, we have
	\be\label{gradient}
	%(\partial_{p_1^*}\alpha, \partial_{p_2^*}\alpha, \partial_{p_3^*}\alpha= 
	\partial_{p_1^*,p_2^*,p_3^*}\alpha=
	\begin{cases}
		\big(2-4p_2^*p_3^* , -1 + 6p_2^*p_3^* -4p_1^*p_3^*, -1+3(p_3^*)^2 + 3(p_2^*)^2 -4p_1^* p_2^* \big)
		& p_1^*\leq p_2^*\leq p_3^*,\\
		\big( 2- 6p_1^*p_3^* + 2 p_2^*p_3^*, -1+ 2p_1^*p_3^*, -1 + 3(p_3^*)^2 -3(p_1^*)^2 + 2p_1^*p_2^* \big)
		& p_2^*<p_1^*<p_3^*,\\
		\big( 2-6(p_1^*)^2+2p_2^*p_3^*, -1 + 2p_1^*p_3^*, -1 + 2p_1^*p_2^*  \big)
		& p_2^*<p_3^*<p_1^*\\
	\end{cases}
	\ee
	and
	\be\label{secondpart}
	(\partial_{p_1^*}^2 \alpha, \partial_{p_2^*}^2\alpha , \partial_{p_3^*}^2\alpha)=
	\begin{cases}
		\big(0 ,   6p_3^* , 6p_3^*\big)
		& p_1^*\leq p_2^*\leq p_3^*,\\
		\big( - 6p_3^* , 0,  6p_3^* \big)
		& p_2^*<p_1^*<p_3^*,\\
		\big( -12p_1^*, 0 , 0  \big)
		& p_2^*<p_3^*<p_1^*.\\
	\end{cases}
	\ee

	From \eqref{3alpha}-\eqref{gradient}, $C^1$ regularity then follows by inspection, 
	comparing values and partial derivatives at boundaries $p_1^*=p_2^*$ and $p_1^*=p_3^*$
	of the different domains of definition.
	Concavity/convexity with respect to individual variables $p_j^*$ then follow by concavity/convexity
	on separate domains of definition, which follows by inspection by \eqref{secondpart} together
	together with $0\leq p_j^*\leq 1$.
	It was shown by direct computation in \cite{CCZ}, on the other hand, 
	that $\alpha$ is not jointly convex in $(p_2^*,p_3^*)$.
\end{proof}

As shown in \cite{CCZ,BLPWZ}, the unique symmetric Nash equilibrium for $\alpha$ is 
\be\label{gutsnash}
(p_1^*,p_2^*,p_3^*)=(1/\sqrt{2}, 1/\sqrt{2}, 1/\sqrt{2}),
\ee
with value $V_N=0$, but the synchronous coalition value is $V_S<0$,
giving a winning outcome for players 2-3 for mixed synchronized strategies.
Our next result shows that there are no winning pure strategies for players 2-3, a first step in 
showing that $V_A=V_N=0$.
(If such strategies did exist, we would have instead $V_A<V_N$, as pure strategies are a
case of asynchronous ones.)

\bl\label{nosaddleprop}
The payoff function $\alpha$ of \eqref{3alpha}, considered as a game between player 1 and a coalition
of players 2-3 has no winning pure strategy solution for players 2-n, that is,
	\be\label{ns}
	\min_{p_2^*,p_3^*}\max_{p_1^*}\alpha(p_1^*,p_2^*,p_3^*)=0.
	\ee
	Moreover, $\max_{p_1}\alpha(p_1,p_2^*,p_3^*)$ has a single local and global minimum,
	at $p_2^*=p_3^*=1/\sqrt{2}$.
\el

\begin{proof}
	{\bf 1.}
	Without loss of generality (by symmetry) take $p_2^*\leq p_3^*$.
	Then, easy calculations show that $\alpha(p_2^*,p_2^*,p_3^*)\geq 0$
	whenever
	\be\label{A}
	(p_2^*)^2 + p_2^*p_3^*\geq 1
	\ee
	and $\alpha(p_3^*,p_2^*,p_3^*)\geq 0$ whenever
	\be\label{B}
	(p_3^*)^2\leq 1/2.
	\ee
	It remains to show that $\alpha(p_1^*,p_2^*,p_3^*)\geq 0$ for some $p_1^*$ when
	\eqref{A}-\eqref{B} both fail, that is, when
	\be\label{nA}
	(p_2^*)^2 + p_2^*p_3^*\leq 1
	\ee
	and
	\be\label{nB}
	(p_3^*)^2\geq 1/2,
	\ee
	to which case we now restrict.

	Computing $\partial_{p_1^*}$ for $p_2^*\leq p_1^*\leq p_3^*$, we have
	\be\label{midpart}
	\partial
	\partial_{p_1^*}\alpha (p_1^*,p_2^*,p_3^*)=2- 6p_1^*p_3^* + 2 p_2^*p_3^*,
	\ee
	hence
	$$
	\partial_{p_1^*}\alpha (p_2^*,p_2^*,p_3^*)=
	2(1- 2 p_2^*p_3^*) > 0
	$$
	or else 
		$$
		\begin{aligned}
			\alpha &= -p_2^*-p_3^*+(p_3^*)^3+3(p_2^*)^2p_3^*\\
			&=
			p_2^*(2p_2^*p_3^*-1 + p_3^* ((p_3^*)^2+(p_2^*)^2-1)\\
			&\geq
			p_2^*(2p_2^*p_3^*-1) + p_3^* (2  p_2^*(p_3^*-1) \geq 0.
		\end{aligned}
		$$

Likewise,
	$$
	\partial_{p_1^*}\alpha (p_3^*,p_2^*,p_3^*)=
	2- 6(p_3^*)^2 + 2 p_2^*p_3^* \leq 3-6(p_3^*)^2 <0
	$$
	by 
	%\eqref{nA} and 
	\eqref{nB}.
	Thus, in the remaining case $p_2p_3<1/2$ there is a maximum with respect to $p_1^*$ in $(p_2^*,p_3^*)$, 
	at which $\partial_{p_1^*}\alpha=0$.
	
	Setting the derivative equal to zero in \eqref{midpart}, we obtain after rearrangement the maximal argument
	\be\label{argmax}
	p_1^*=(1/3)(1/p_3^*+ p_2^*),
	\ee
	or, equivalently,
	\be\label{equiv}
	 3p_1^* p_3^*= 1+ p_2^*p_3^*.
	\ee

	Using \eqref{equiv} to simplify \eqref{3alpha}(iii), we find at this maximal point that
	\ba
	\alpha(p_1^*,p_2^*,p_3^*)&=
	2p_1^*-p_2^*-p_3^*+(p_3^*)^3-3(p_1^*)^2p_3^* + 2p_1^*p_2^*p_3^*\\
	&= 2p_1^*-p_2^*-p_3^* + (p_3^*)^3 -p_1^*(1+ p_2^*p_3^*) + (2/3)p_2^*(1+ p_2^*p_3^*)\\
	&= p_1^*-(2/3)p_2^*-p_3^* + (p_3^*)^3 + (1/3)(p_2^*)^2 p_3^*.
	\ea
	Thus, multiplying by $3p_3^*$ and applying \eqref{equiv} again, we have
	\ba
	3p_3^* \alpha(p_1^*,p_2^*,p_3^*)&= 3 p_1^*p_3^* -2p_2^*p_3^*  + 3(p_3^*)^4 + (p_2^*)^2 (p_3^*)^2\\
	 &= (1+p_2^*p_3^*) -2p_2^*p_3^*  + 3(p_3^*)^4 + (p_2^*)^2 (p_3^*)^2\\
	 &= (1-p_2^*p_3^*) + 3(p_3^*)^4 + (p_2^*)^2 (p_3^*)^2,
	\ea
	whence, using \eqref{nB}, 
	$
	3p_3^* \alpha(p_1^*,p_2^*,p_3^*)\geq
	 (p_2^*)^2 + 3(p_3^*)^4 + (p_2^*)^2 (p_3^*)^2\geq 0.
	 $

	 Combining, we have $\min_{p_2^*,p_3^*}\max_{p_1^*}\alpha(p_1^*,p_2^*,p_3^*)\geq 0$.
	 But, by \eqref{gutsnash}, the Nash equilibrium strategy $p_2^*=p_3^*=1/\sqrt{2}$
	 guarantees a return to player 1 of $\leq 0$, whence we may conclude \eqref{ns}.

	 {\bf 2.} (Unique local minimum) As $\alpha$ is $C^1$ on its domain $[0,1]^3$, 
	 and $\alpha_{p_1p_1}<0$ unless $p_1^*\leq p_2^*, p_3^*$, implying uniqueness of
	 $argmin \alpha(\cdot, p_2^*,p_3^*)$,
	 any interior local minimum in $(p_2^*,p_3^*)$
	 of $\max_{p_1^*}\alpha(p_1^*,p_2^*,p_3^*)$ by Proposition \ref{critprop}
	 must be a critical (saddle) point, except possibly in the case 
	 $p_1^*\leq p_2^*, p_3^*$, in which $0=\alpha_{p_1}= 2-4p_2^*p_3^*$ implies $p_2^*p_3^*=1/2$,
	 and 
	 $$
	 \alpha= -p_2^*-p_3^*+(p_3^*)^3 +3(p_2^*)^2p_3^*= 1/4p_3^*-p_3^*+(p_3^*)^3.
	 $$
	 Differentiating along the path $p_2^*p_3^*=1/2$, parametrized by $p_3^*$,
	 we thus have 
	 $$
	 \begin{aligned}
		 d\alpha/dp_3^* &= -1/4(p_3^*)^2 -1 + 3(p_3^*)^2\\
		 &=  -(p_2^*)^2 -1 + 3/4(p_2^*)^2,
	 \end{aligned}
	 $$
	 which, by $(p_2^*)^2\leq 1/2$ is greater than or equal to $-1/2 -1 + 3/2=0$,
	 with equality only if $p_2=1/\sqrt{2}$. Thus, the only possible local minimum 
	 is at $p_2^*=p_3^*=1/\sqrt{2}$, which is the global minimum.

	 As remaining possible local minima are critical points, 
	 we need thus to check separately only for interior critical points and for local minima on the
	 boundaries $p_j^*=0,1$.

	 {\bf 2(a).} (Unique interior critical point)
	 From \eqref{gradient}, we readily find that the unique critical point
	 occurs at $p_1^*=p_2^*=p_3^*=1/\sqrt{2}$.
	 We examine each of cases \eqref{gradient}(i)-(iii) in turn.

	 {\it Case \eqref{gradient}(i)} ($p_1^*\leq p_2^*\leq p_3^*$)
	 Setting $\partial_{p_1^*}\alpha=0$, we have $1=2p_2^*p_3^*$. From $\partial_{p_2^*}\alpha=0$,
	 we then have $0=-1+6p_2^*p_3^*-4p_1^*p_3^*= 2 -4p_1^*p_3^*$, or $1=2p_1^*p_3^*$.
	 Combining gives $p_3^*=p_2^*$, from which we obtain $p_1^*=p_2^*=p_3^*$ by $p_1^*\leq p_2^*\leq p_3^*$,
	 and thus $p_j^*=1/\sqrt{2}$ for $j=1,\dots,3$.
	 
	 {\it Case \eqref{gradient}(ii)} ($p_2^*\leq p_1^*\leq p_3^*$)
	 From $\partial_{p_2^*}\alpha=0$ we obtain $1=2p_1^*p_3^*$, while from $\partial_{p_1^*}\alpha=0$
	 we obtain
	 $0=2-6p_1^*p_3^*+2p_2^*p_3^*=-1 + 2p_2^*p_3^*$, or $2p_2^*p_3^*=1$.
	 Combining, we have $p_1^*=p_2^*$, and thus $p_1^*=p_2^*=p_3^*$, giving again $p_j^*=1/\sqrt{2}$.
	 
	 {\it Case \eqref{gradient}(iii)} ($p_2^*\leq p_3^*\leq p_1^*$)
	 From $\partial_{p_3^*}\alpha=0$ we obtain $1=2p_1^*p_2^*$, while from $\partial_{p_2^*}\alpha=0$
	 we obtain $2p_1^*p_3^*=1$. Combining, we have $p_2^*=p_3^*$.
	 From $\partial_{p_1^*}\alpha=0$, finally, we obtain $0=2-6(p_1^*)^2+ 2p_2^*p_3^*= 3-6(p_1^*)^2$,
	 giving $p_1^*=1/\sqrt{2}$, and thus $p_2^*=p_3^*=1/\sqrt{2}$.

	 {\bf 2(b).} (No boundary local minima)
	 It remains to check the various boundaries $p_j^*=0,1$ and verify that no local minimum
	 in $(p_2^*,p_3^*)$ can occur there. We omit these straightforward computations.

	 {\it Case $p_1^*=0$, \eqref{gradient}(i)} ($p_1^*\leq p_2^*\leq p_3^*$)
	 For a local minimum, there must hold $0\geq \partial_{p_1^*}\alpha =2-p_2^*p_3^* $ and 
	 also $0=\partial_{p_2^*}\alpha=-1+6p_2^*p_3^*-4p_1^*p_3^*= -1+6p_2^*p_3^*$, or $ p_2^*p_3^*=1/6$,
	 a contradiction.

	 {\it Case $p_1^*=1$, \eqref{gradient}(iii)} ($p_2^*\leq p_3^*\leq p_1^*$)
	 For a local minimum, there must hold $0\leq \partial_{p_1^*}\alpha =2-p_2^*p_3^* $ and 
	 also $0=\partial_{p_2^*}\alpha= \partial_{p_3^*}\alpha$. The latter two equalities give
	 $p_2^*=p_3^*$, hence the first inequality becomes $0\leq -4 + 2(p_2^*)^2\leq -2$, a contradiction.

	 {\it Case $p_2^*=0$, $p_1^*\neq 0$, \eqref{gradient}(ii)} ($p_2^*\leq p_1^*\leq p_3^*$)
	 For a local minimum, $0\geq \partial_{p_2^*}\alpha =2-p_1^*p_3^* $ and 
	 $0=\partial_{p_1^*}\alpha= 1-6p_1^*p_3^*$, or $p_1^*p_3^*=1/6$, a contradiction.

	 {\it Case $p_2^*=0$, $p_1^*\neq 0$, \eqref{gradient}(iii)} ($p_2^*\leq p_3^*\leq p_1^*$)
	 For a local minimum, $0\geq \partial_{p_2^*}\alpha =2-p_1^*p_3^* $ and 
	 $0=\partial_{p_3^*}\alpha= 1-2p_1^*p_2^*=1$, a contradiction.

	 {\it Case $p_3^*=1$, $p_1^*\neq 1$, \eqref{gradient}(i)} ($p_1^*\leq p_2^*\leq p_3^*$)
	 For a local minimum, $0\leq \partial_{p_3^*}\alpha =2-p_1^*p_3^* $ and 
	 $0=\partial_{p_1^*} \alpha =\partial_{p_2^*} \alpha$. But, the last two already give 
	 $p_j^*=1/\sqrt{2}$ for all $j=1,\dots,3$, a contradiction.

	 {\it Case $p_3^*=1$, $p_1^*\neq 1$, \eqref{gradient}(ii)} ($p_2^*\leq p_1^*\leq p_3^*$)
	 Again, the requirements $0=\partial_{p_1^*} \alpha =\partial_{p_2^*} \alpha$
	 give $p_j^*=1/\sqrt{2}$ for all $j=1,\dots,3$, a contradiction.

	 Combining the results of the above cases, we find that there exist no local minima with respect
	 to $p_2^*$, $p_3^*$ on the boundary, completing the proof.
\end{proof}

\br\label{altNash}
The second half of the proof of Lemma \ref{nosaddleprop} gives in passing
an alternative proof that $(1/\sqrt{2},1/\sqrt{2}, 1/\sqrt{2})$ is a (unique)
symmetric Nash equilibrium, as a Nash equilibrium is necessarily a critical point of $\alpha$.
\er

\br\label{Nashrmk}
For the $n$-player game, the symmetric Nash equilibrium is $ p_j^*=(1/2^{1/(n-1)})$.
This may be recognized as the player 1 value $p_1^*$ for which there is $1/2$ probability 
that one of players 2-$n$ have a higher hand.
The explanation for this convenient rule of thumb
is that marginal return for deviation of player one from strategy $p_1^*$ is then zero,
a necessary condition for optimality.
For, it is $-(n-1)$ in the case that players $2$-$n$ have hands less than $p_1^*$ and drop,
since increase in $p_1$ then leads (up to $O(p_1-p_1^*)$) 
to loss of $p_1-p_1^*$ times winnings $n-2$ lost due to dropping instead of holding, so marginal 
return of 
$$
1/(p_1-p_1^*) \times (p_1-p_1^*) \times -(n-1)=-(n-1)
$$
in the limit as $p_1\to p_1^*$.
But, it is $+(n-1)$ in the case that one or more of players $2$-$n$ have hands greater than $p_1^*$, 
since increase in $p_1$ then leads to $p_1-p_1^*$ times losings avoided due to dropping instead of holding,
which are $(n-1)$ independent of the number of players holding greater than or equal to one, with remaining
future possible earnings due to growth of the pot independent of whether player one holds or drops.
As these cases are equally likely at probablility $1/2$, the total marginal return thus averages to zero
as claimed.
This point of view perhaps illuminates somewhat the proof of optimality in \cite{CCZ}.

Similarly, an interesting Guts variation is the ``Weenie rule,'' in which if all players drop then the player with
highest hand must match the pot.
In this case the marginal return for increase of player one from a given common strategy $p_*$ is
$(n-1)$ as usual in the case that one or more of players $2$-$n$ have hands greater than 
$p^*$, but $-2(n-1)$ in the case that none of players $2$-$n$ have hands greater than 
$p^*$. Thus, total marginal return is zero when the probability of the latter is twice that of the former,
i.e., the probability that players $2$-$n$ have hands $\leq p^*$ is $1/3$, or,  
equivalently, the Nash equilibrium strategy is
$(p^*)^{n-1}=1/3$. This recovers and further illuminates the optimality result of \cite[Appendix B]{BLPWZ}.
\er

%%%%%%%%%%%%%%%%%%%%%%%%%%%%%%%%%%%%%%%%%%%%%%%%%%%%%%%%%%
\subsection{Special properties of continuous games}\label{s:contfrag}
We consider now the general case of an n-player continuous game, 
with payoff matrix $\Psi(x_1,\dots,x_n)$, $x_j\in [0,1]$.
A mixed strategy for such a game for player j consists of a probability distribution $\mu_j$ on $[0,1]$.
A mixed strategy for a synchronous coalition of players 2-n consists of a probability distribution $\nu$
on $[0,1]^{n-1}$, a mixed strategy for an asynchronous coalition a product distribution
$\nu(x_2, \dots, x_n)=\mu_2(x_2)\cdots \mu_n(x_n)$.

The continuous formulation imposes additional structure on $\Psi$ as compared to the payoff matrix
$P$ of the discrete case, allowing us to say more in certain special cases.
The first is somewhat hypothetical, as we do not have an interesting example of this case.

\begin{proposition}\label{minimaxprop}
	If $\Psi$ is concave in $x_1$ and convex in $(x_2, \dots, x_n)$, then there exist pure strategies
	$x_1^*$ for player 1 and $x_2^*,\dots, x_n^*$ for players 2-n forcing the optimal value $V_S$,
	whence $V_A=V_S$.
	If also $\Psi$ is symmetric, then there is an optimal repeated pure strategy for players 2-n, whence
	$V_S=V_A=V_N=0$.
\end{proposition}

\begin{proof}
	The first assertion is the classical minimax theorem of von Neumann \cite{vN}.
	The second is evident from the fact that pure strategies are a subset of strategies of product type.
	The third follows from the fact, by symmetry, that every permutation of an optimal pure strategy 
	for players 2-n is also optimal, whence, so also is the mixed synchronous strategy consisting of
	the average of these permutations.  But, by Jensen's inequality, using convexity in $x_2,\dots, x_n$,
	we have that the average of the payoffs is greater than or equal the payoff of the average,
	and so the averaged strategy must be optimal as well. But, this is a repeated pure strategy, 
	giving the result.
\end{proof}

The second is of more interest, applying directly to continuous Guts Poker.

\begin{proposition}\label{indprop}
	If $\Psi$ is individually convex in each of  $x_2, \dots, x_n$, then there exist pure strategies
	$x_2^*,\dots, x_n^*$ for players 2-n forcing the optimal value $V_A$, i.e.,
	\be\label{VAminimax}
	V_A= \min_{x_2^*, \dots, x_n^*} \max_{x_1^*} \Psi(x_1^*, x_2^*, \dots, x_n^*).
	\ee
\end{proposition}

\begin{proof}
	Applying Jensen's inequality in turn in each of the coordinates $x_2,\dots, x_n$, we find
	that any mixed product strategy for players 2-n is majorized by a pure strategy, given by
	the expectations of each $x_j$ under the associated probability distributions,
	whence the latter gives an optimal pure strategy for the asynchronous coalition game.
	The conclusion \eqref{VAminimax} then follows by the fundamental theorem of two-player games,
	considering the coalition of players $2$-$n$ as a single player.
\end{proof}

Proposition \ref{indprop} reduces the problem of finding $V_A$ from an infinite-dimensional to a finite-dimensional
optimization problem, the latter treatable in principle by standard Calculus.

\subsection{Conclusion}\label{s:conc}
Combining the results of the previous two subsections, we have the following definitive result,
answering in the affirmative the main open problem posed in \cite{BLPWZ}.

\begin{corollary}\label{gutscor}
	For 3-player continuous Guts Poker, $V_S< V_A=V_S=0$.
\end{corollary}

\begin{proof}
	This is an immediate consequence of Proposition \ref{indprop} and Lemmas
\ref{C1lem} and \ref{nosaddleprop}.
\end{proof}

\br\label{conjrmk}
We conjecture that the result of Corollary \ref{gutscor} holds true for the  $n$-player game as well.
However, we have not carried out the necessary analysis of the $n$-player payoff function needed
to conclude. This would be a very interesting problem for further analytical investigation.
\er

%%%%%%%%%%%%%%%%%%%%%%%%%%%%%%%%%%%%%%%%%%%%%%%%%%%%%%%%%%

\subsection{The synchronous coalition revisited}\label{s:revisit}
Remarkably, returning to the synchronous coalition problem, we may obtain
an explicit value for $V_S$ by the same set of techniques.
Analogous to Proposition \ref{indprop}, we have the following complementary result applying to synchronous
coalitions.
\begin{proposition}\label{maximinprop}
	Let $\Psi$ as in Proposition \ref{indprop} denote a continuous payoff function.
	If $\Psi$ is concave in $x_1$, then there exists an optimal pure strategy
	$x_1^*$ for player 1 forcing the optimal value $V_S$, i.e.,
	\be\label{VSmaximin}
	V_S= \max_{x_1^*}\min_{x_2^*, \dots, x_n^*} \Psi(x_1^*, x_2^*, \dots, x_n^*).
	\ee
\end{proposition}

\begin{proof}
	The first assertion follows by Jensen's inequality as in the proof of Proposition
\ref{minimaxprop}, whence the second follows by the fundamental theorem of two-player games
	(as usual, considering the coalition of players $2$-$n$ as a single player opposing
	player 1).
\end{proof}

As Proposition \ref{indprop} did for $V_A$,
Proposition \ref{maximinprop} reduces the infinite-dimensional problem of finding $V_S$ to a Calculus
problem in an $n$-dimensional domain.
We recall from \cite[Prop. 6.8]{CCZ} and surrounding discussion the following result.

\begin{proposition} \label{3br}
For continuous 3-player Guts, the best response function is given by 
	  \be\label{R}
	  R(p_1^*):=\min_{p_2^*,p_3^*}\alpha(p^*_1,p^*_2,p^*_3)= \min \{\alpha_a(p_1^*), \alpha_b(p_1^*)\},
\ee
	  with respective values
	  \ba\label{3breq}
	  \alpha_a&=-\frac{1}{27}\Big((4(p_1^*)^2+6)^{\frac{3}{2}}+8(p_1^*)^3-36p_1^*\Big),\\
	  \alpha_b &=-\frac{2}{27}\Big((9(p_1^*)^2+3)^{\frac{3}{2}}-27p_1^*\Big) ,
	  \ea
	achieved at $(p_2^*, p_3^*)=(p_3^a,p_3^a)=((\sqrt{4(p_1^*)^2+6}-2p_1^*)^{-1}, (\sqrt{4(p_1^*)^2+6}-2p_1^*)^{-1})$ and
	  $(p_2^*, p_3^*)=(0,p_3^b)=\Big(0, \sqrt{\frac{3(p_1^*)^2+1}{3}}\Big)$.
	  Here $\alpha_1<0$ except at $p_1^*=1/\sqrt{2}$, where it is zero, and $\alpha_b(1/\sqrt{2})<0$.
\end{proposition}

\begin{corollary}\label{VScor}
	For continuous guts poker the one-shot game has value 
	\be\label{xvalue}
	V_S=\max_{p_1^*}\min\{\alpha_a(p_1^*), \alpha_b(p_1^*)\}<0.
	\ee
\end{corollary}

\begin{proof}
	Noting for continuous guts poker that $\alpha$ is concave with respect to $x_1$, we have by
	Proposition \ref{maximinprop} that
	\be\label{straightup}
	V_S=\max_{p_1^*} \min_{p_2^*,p_3^*}\alpha(p_1^*,p_2^*,p_3^*)= \max_{p_1^*}R(p_1^*),
	\ee
	from which the result then follows by \eqref{R}.
\end{proof}

\br\label{oscrmk}
The advantage of \eqref{xvalue} over \eqref{straightup} is illustrated by
Figure \ref{oscfig}, below, in which \eqref{straightup} is used directly, treating the inner
loop as a nonconvex numerical minimization problem using the BFGS algorithm as described in
Section \ref{s:num}.
As we see, the presence of nearby local minima corresponding to $a(\cdot)$ and $b(\cdot)$ in \eqref{xvalue}
results in unwanted oscillation between the true global minimum and its nearby local neighbor.
\er

\begin{figure}
\centering
\includegraphics[scale=0.5]{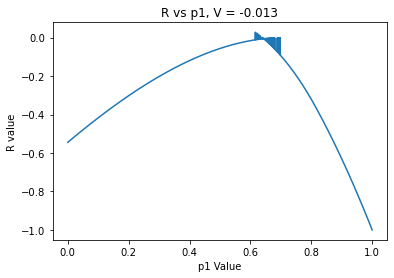}
\caption{Numerical evaluation of inner loop of Maximin, illustrating
oscillation between local minima; cf. Figure \ref{negativefig}.  }
\label{oscfig}
 \end{figure}

\subsubsection{Structure of optimizers}\label{s:struct}
As a side-benefit, the above analysis explains the remarkably simple structure of optimal strategies
observed for the 3-player synchronous coalition game in \cite{BLPWZ} of a pure strategy for player 1
and a mixture of just two pure strategy pairs (and their symmetric counterparts with roles of players 2 and
3 exchanged) for the player 2-3 coalition, of form $A=(0, q_1)$ and $B=(q_2,q_2)$.  
The first observation we have already established by 
a Jensen theorem argument showing that pure strategies are optimal for player 1.  The second now
follows by a closer look at Proposition \ref{3br} and Corollary \ref{VScor}, which characterize best response
for players 2-3 as one of the two strategies $A$ and $B$ of this form, and the optimal strategy for player 1
as the value of $p_1^*$ at which the minimum of these two responses is maximized.
But, this can be seen to occur at a point where the two responses give equal value, and all other strategies
(by the analysis of the proof of Proposition \ref{3br} in \cite{CCZ}) a greater one;
see Figure \ref{negativefig}.
Thus, the optimal strategy for players 2-3 must consist of a mixture of only these two strategies, for that
is the only possiblity that is optimal against the optimal $p_1^*$.

  \begin{figure}
        \centering
		  \includegraphics[scale=0.25]{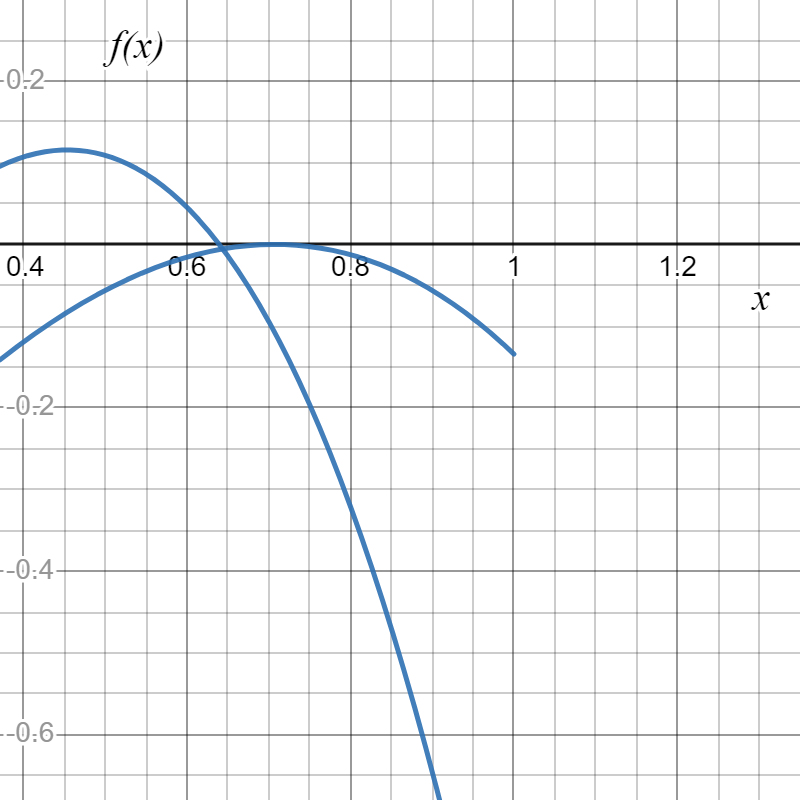}
		 	\caption{
				Blowup of graph plotting best response $\alpha_a$ and $\alpha_b$ vs. $p_1^*$.
				Maximin occurs at $p_1^*\approx 0.6437$, $\alpha\approx -0.0056$, at
				intersection of $\alpha_a$ and $\alpha_b$.
				}
	\label{negativefig}
  \end{figure}

Indeed, we can give the following complete description.

\begin{proposition}\label{odesc}
For continuous 3-player Guts with synchronous coalition, the optimal one-shot strategy for
	player 1 is the pure strategy $p_1^o=argmax R$, while an optimum strategy for players 2-3 is
	\be\label{23opt}
	(p_2^o,p_3^0)=\begin{cases}
		(p_3^a(p_1^o,p_3^1(p_1^0))& probability \, y,\\
		(0,p_3^b(p_1^0))& probability \, (1-y) 
		\end{cases}
	\ee
	with
	\be\label{yval}
	y= 
	\frac {|\alpha_b'(p_1^0)|} {|\alpha_a'(p_1^0| + |\alpha_b'(p_1^0| }.
	\ee
\end{proposition}

\begin{proof}
	The description of the optimal strategy for player 1 follows from the Jensen theorem argument showing
	that the optimal strategy is of pure type, together with the definition of $R$.
Meanwhile, by the discussion just above, together with numerics validating the assumptions made therein, as displayed
	in Figure \ref{negativefig}, we have that the optimal strategy for players 2-3 is of form \eqref{23opt}.
	Noting that, by concavity in $p_1^*$, the expected payoff 
	$$
	\psi(p_1^*)= y\alpha(p_1^*, p_3^a(p_1^o),p_3^1(p_1^0)) + (1-y) \alpha(p_1^*,0,p_3^b(p_1^0))
	$$
	for $\alpha(p_1^*, p_2^0, p_3^0)$
	is concave, with value equal to $V_S$ at $p_1^*=p_1^o$, we find that it is less than or equal to $V_S$
	for all $p_1^*$, hence optimal, if and only if $\psi'(p_1^0)=0$, or
	\be\label{keyy}
	0=y \partial_{p_1} \alpha(p_1^*, p_3^a(p_1^o,p_3^1(p_1^0)) + (1-y) \partial_{p_1} \alpha(p_1^*,0,p_3^b(p_1^0))
	\ee
	at $p_1^*=p_1^0$.  Noting that 
	$\partial_{p_1} \alpha(p_1^*, p_3^a(p_1^o,p_3^1(p_1^0))=\alpha_a'(p_1^0)$
	and $\partial_{p_1} \alpha(p_1^*,0,p_3^b(p_1^0))=\alpha_b'(p_1^0)$ by minimality of $\alpha_a$ and $\alpha_b$
	with respect to $p_3$ at $p_1^0$, with $\alpha_a'(p_1^0)$ and $\alpha_b'(p_1^0$ of opposite signs,
	we obtain the result \eqref{yval} by substituting these facts into \eqref{keyy} and solving for $y$.
\end{proof}

In Figure \ref{tanfig}, we illustrate the construction of the optimal coalition solution for players 2-3
described in Proposition \ref{odesc}, displaying tangent lines to $\alpha_a$ and $\alpha_b$.
The values for $p_3^a$, $p_3^b$ and $y$ so obtained are very close to the values
$0.68$, $0.86$, and $0.86$ obtained numerically by fictitious play in \cite[\S 5]{BLPWZ}.
In Figure \ref{optfig}, we show the resulting return for player 1 against this optimal strategy
as a function of player 1 strategy $p_1^*$, which can be seen to be strictly negative.

  \begin{figure}
        \centering
		  \includegraphics[scale=0.25]{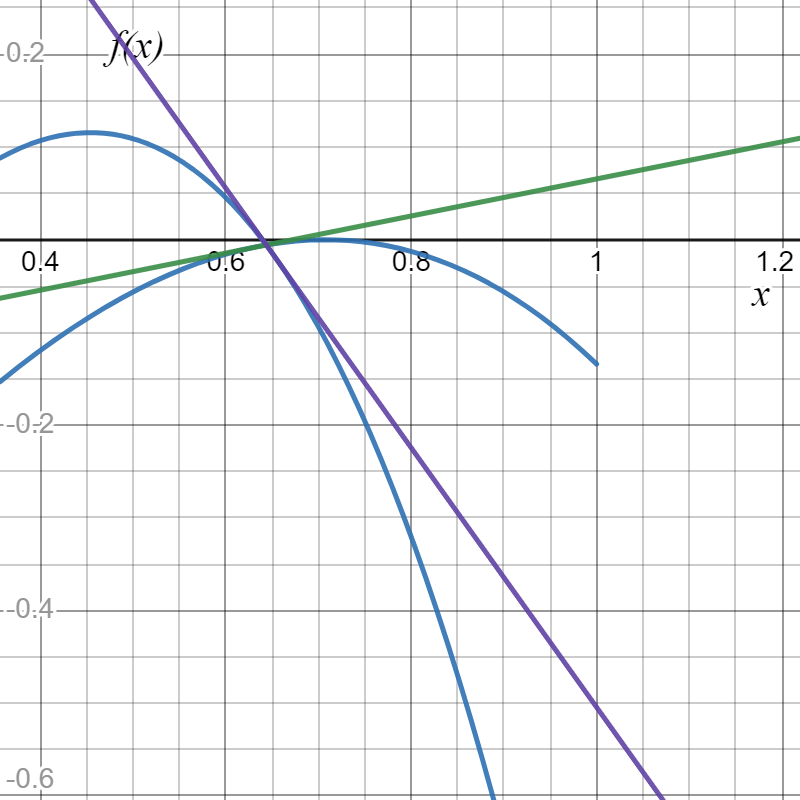}
		 	\caption{
				Blowup of graph plotting best response $\alpha_a$ and $\alpha_b$ vs. $p_1^*$.
				with tangent lines at crossing points showing $\alpha_a'\approx .02$,
				$\alpha_b'\approx -1.4$ at maximin occurring at
				intersection of $\alpha_a$ and $\alpha_b$, giving probability
				$y\approx .875$ for strategy $(p_a^3,p_a^3)$ and $(1-y)\approx .125$
				for strategy $(0,p_b^3)$, where, by our formulae, $p_3^a\approx .6764$ and
				$p_3^b\approx .864$.
				}
	\label{tanfig}
  \end{figure}

  \begin{figure}
        \centering
		  \includegraphics[scale=0.25]{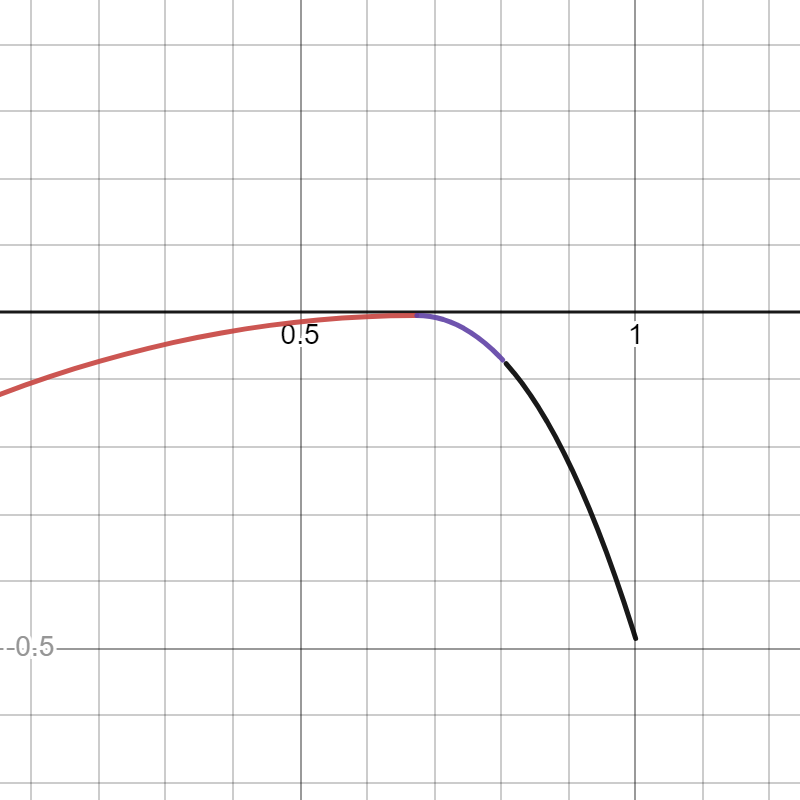}
		 	\caption{
				Payoff for player 1 against optimal player 2-3 strategy, plotted vs.  $p_1^*$.
				}
	\label{optfig}
  \end{figure}

  \br\label{usefulrmk}
A useful comment about our construction of optimal player 2-3 strategies
for synchronous coalition guts: the heuristic argument
for this depends on an infinite-dimensional fundamental theorem of games, which is a little delicate and we don't discuss, asserting existence of the minimax. We then conclude easily that it could involve only the two local mininma
at $p_1^o$ where they coincide. However, building on this intuition, we actually construct one, and verify 
rigorously the optimality, thus sidestepping the need for any abstract infinite-dimensional theory.
We note that this infinite-dimensional theory is nontrivial, and is proved with appropriate additional assumptions
on the infinite-dimensional model, which may or may not be satisfied for continuous guts.
\er

\subsubsection{The recursive game}\label{s:recsync}
The recursive version of continuous guts is based on the value function $T(V)$, $V\in \R$,
defined as the coalition value $V_S$ for the game $\alpha+ V \beta$, with
\be\label{beta}
\beta= 2-p_1^*-p_2^*-p_3^*+2p_1^*p_2^*p_3^* .
\ee
Namely, the value forceable over repeated rounds of play by a coalition of players 2-3 may be seen
to be equal to
\be\label{limval}
V_*:= \lim_{n\to \infty} V_S^n:= T^n(0),
\ee
where $V_S^n$ is a nonincreasing sequence, with $V_1$ equal to the value $V_S$ of the one-shot game 
discussed just above.

Noting that $\beta$ is linear in $p_1^*$, so that $\alpha + V \beta $ is concave with respect to $p_1^*$
for any $V$, we have the formula
\be\label{recform}
T(V)= \max_{p_1^*}\min_{p_2^*,p_3^*}(\alpha + V\beta)(p_1^*,p_2^*,p_3^*)
\ee
analogous to the one-shot case.
In principle, we could carry out a best response analysis for the recursive game like that of Proposition
\ref{3br} to determine sharp estimates for $V_*$ directly, without numerical optimization.
See Appendix \ref{s:alphabeta} for further discussion, in particular
the computation of $T(V)$ carried out in Proposition \ref{r3br}.

%%%%%%%%%%%%%%%%%%%%%%%%%%%%

\subsection{Discretization}\label{s:discretization}
As in \cite{BLPWZ}, one may approach the numerical study of continuous Guts Poker by
discretizing the pure strategy space $(p_1^*,p_2^*,p_3^*)\in [0,1]^3$, considering equally spaced
values $p_j^*= 0, 1/N, \dots, 1-1/N$ for an $N$-point mesh.
This determines an $N\times N\times N$ discrete three player game, that is also symmetric, zero sum,
differing by at most $(1/N)\max|\nabla \alpha|$ in values $V_A$ and $V_S$.
Amusingly, this game corresponds to a simplified version of discrete Guts Poker,
played with hands consisting of a single card
taking values $1,\dots, N$ with equal likelihood and chosen with replacement.
Though amenable to numerics, the discrete game does not necessarily share the favorable property
of the continuous game that optimal strategies be of pure type, nor of unique local minima for the
minmax problem; at most we can say that they are nearby their continuous analogs. 
As these games are large and of a rather special type induced by closeness to the special structure of continuous
Guts, and because we possess already a complete solution in the continuous case,
we shall not carry out a detailed numerical study of the discretized case.
We remark only that nonsystematic trials suggest that both maximin and minimax problems seem
amenable to all of the numerical methods discussed below, consistently reaching the global
minimum, with, however, some unwanted ``spreading'' of
strategies away from the pure strategies that are optimal for the continuous case.
(The exception regarding spreading is Fictitious play as studied in \cite{BLPWZ}, which does stay near 
pure strategies but which is applicable in full strength only for the synchronous coalition case.) 
This good behavior can be explained by the fact that the nearby asynchronous optimization problems 
for continuous guts possess only a single local minimum, which is also a global one: the maximin problem since it is
convex, and the (nonconvex) minimax problem by the special features observed above.
%END NEW

%%%%%%%%%%%%%%%%%%%%%%%%%%%%%%%%%%%%%%%%%%%%%%%%%%%%%%%%%%
\section{Numerical optimization}\label{s:num}
In the previous sections, we have answered to various extents the analytical problems posed in \cite{BLPWZ}
of when and for what types of systems do $V_S=V_A<N_S$, $V_S<V_A<N_S$, or $V_S<V_A=V_N$?
In particular, we have answered in the affirmative the main open problem whether $V_A=V_N=0$ 
for continuous Guts Poker, as conjectured in \cite{BLPWZ}.

We now turn to the further challenge posed in \cite{BLPWZ} of efficient numerical approximation of
$V_A$ for general systems, not necessarily possessing any special structure by which this may be
deduced analytically. For this task, our explicitly soluble
examples of Rock-Paper-Scissors, and randomly generated matrix games
serve as useful benchmarks.

\subsection{Minimax vs. Maximin}\label{s:mm}
Consider a general $m\times m\times m$ three-player symmetric zero-sum game, with players 2-3 acting as a coalition to minimize the payoff to player 1. Denoting the payoff for pure strategies $i$, $j$, $k$ as $P_{ijk}$, and mixed strategies by
probability vectors $x, y, z\in \R^{m+}$, $\sum_i x_i=\sum_j y_j=\sum_k z_k=1$, we distinguish the 
{\it Maximin} 
\be\label{maxi}
\max_{x} \min_{y,z} \sum_{i,j,k} x_iy_j z_k P_{ijk}= \max_{x} \min_{j,k} \sum_{i} x_i  P_{ijk}=V_S
\ee
and {\it Minimax} 
\be\label{mini}
\min_{y,z} \max_{x} \sum_{i,j,k} x_iy_j z_k P_{ijk}= \min_{y,z} \max_{i} \sum_{j,k} y_j z_k P_{ijk}=V_A
\ee
problems determining $V_S$ and $V_A$, respectively.

The former, a standard 2-player game, may be determined by the simplex method \cite{D}
or as in \cite{BLPWZ} by the method of fictitious play \cite{B,R}, an iterative scheme in which the two teams
(player 1 and players 2-3) play a series of games in which each plays the best response against the other's
historical empirical strategy distribution.
As a linear programming problem, it could also in principle be treated also by modern interior point methods,
or, more generally, any of the powerful techniques developed for convex numerical optimization;
see \cite{Ne,CVX} and references therein.
A natural question is whether and in what context one of these alternative choices might offer improved 
performance as compared to fictitious play.

The latter, Minimax problem is more complicated numerically.
For, $ \max_{i} \sum_{j,k} y_j z_k P_{ijk}$ is not necessarily convex, making this a nonconvex optimization
problem on the boundary of current investigations.
In particular, it is not clear how the game-theoretic origins of the problem might be exploited to
reduce computation, other than the replacement of the inner optimization over $x$ with maximum in $i$,
already recorded in \eqref{mini}.
So, for the moment, available techniques are those for general nonconvex optimization problems such as originate
in machine learning and elsewhere:
for example, quasi-Newton techniques such as the popular Broyden-Fletcher-Goldfarb-Shanno (BFGS)
algorithm \cite{Br,F1,F2,G,S} supported in SciPy \cite{SPy}, or slower but more reliable first-order
gradient descent.
To avoid trapping in local minima, one may add annealing or ``basin-hopping'' \cite{WD};
indeed, the Python-supported algorithms we test here generally include some version of
these as built-in features.

Additionally we explore the question of the optimal implementation of constraints.  The BFGS method can handle constraints only in the form of bounds on individual components of the strategy vectors.  While we can adjust  our problem by adding a penalty, it is more natural for the restriction to be expressed as an equality ($|x|=1$).  For this reason, we additionally experiment with the various continuation type methods in SciPy, particularly the Sequential Least SQuares Programming (SLSQP) method, which allows for much more general linear constraints.  The details for our numerical protocol can be found in Appendix \ref{s:numpro}.

\subsection{The Maximin problem}\label{s:maximin}
Numerical approximation for the basic two-player game problem represented by \eqref{maxi} has been well studied
in the literature.  See, for example, the discussion of \cite{GHPS1,GHPS2} motivated by
Texas Hold'em Poker, and references therein.
The conclusion is that for large linear programming problems such as occur for standard Poker games,
standard interior point programming methods involve memory requirements that are prohibitive.
Progress for such games has proceeded instead mainly by ``multi-grid'' type methods, iterating over a series of 
approximate problems, and/or smoothing methods as described in \cite{GHPS1,GHPS2}.
Here, guided by these previous results, we investigate the performance of various methods on the maximin problem 
for our example problems of Odds-Evens, Rock-Paper-Scissors, and discretized continuous guts poker, as well as 
for general randomly chosen large games.

With an eye toward the Minimax problem of our main interest, we mainly restrict ourselves to techniques 
available for general, nonconvex nonlinear optimization problems, consisting of (BFGS) and variants supported
in NashPy and SciPy, with and without smoothing. 
As benchmarks, we include also computations by the simplex method and the method of fictitious play. 

\subsubsection{Odds-Evens and Rock-Paper-Scissors}\label{s:minmaxRPS}
We start with the Rock-Paper-Scissors (OMI) and (OMO) games, with payoffs
\be\label{RPSrptOMI}
\Psi(x,y,z)=2 y\cdot z - x\cdot(y+z)
\ee
and 
\be\label{tildePsi}
\tilde \Psi(x,y,z):=-\Psi(x,y,z),
\ee
respectively, where $x, y, z\in \R^{3}$ are probability vectors.
As noted in Remark \ref{maximinrmk}, one may compute directly the innner minimization loop,
optimized at extreme values $y_i=\delta_i^k$ and $z_i=\delta_i^j$ for some $j,k$, to obtain
objective functions
\be\label{minOMI}
\Phi(x):= \min_{y,z}\Psi(x,y,z)= \min_j (x_j) -1
\ee
and
\be\label{minOMO}
\tilde \Phi(x):= \min_{y,z}\tilde \Psi(x,y,z)= 2(\min_j (x_j) -1) 
\ee
to be maximized in $x$.

Noting that the Odds-Evens games can be viewed as restrictions to $x_3=y_3=z_3=0$ of the Rock-Paper-Scissors
game, we find that \eqref{RPSrptOMI} and \eqref{tildePsi} hold here too, whence, computing the inner minimization
loop by hand, we obtain objective functions
\be\label{minOMIe}
\Phi(x):= \min_{y,z}\Psi(x,y,z)=0 - x_1-x_2 \equiv -1, \qquad x,y,z\in \R^2
\ee
for (OMI), and (as in the previous case)
\be\label{minOMOe}
\tilde \Phi(x):= \min_{y,z}\tilde \Psi(x,y,z)= 2(\min_j (x_j) -1), \qquad x,y,z\in \R^2
\ee
for (OMO).

For a general symmetric $N\times N\times N$ game, the inner loop can be carried out numerically at cost
$O(N^2)$ by cycling through the possible pure strategy combinations for players 2 and 3, so this is
essentially the same as the direct evaluation above.  However, analytical derivatives are not generally
available and so we will not use methods that require these.

The main point is that the exact formulae above can be used for forensic evaluation/insight into results.
For example, the formula \eqref{minOMIe} shows that the (OMI) problem for odds-evens is degenerate and
independent of the strategy of player one; hence it is not interesting to carry out.
The formula \eqref{minOMOe} shows that the (OMO) problem for odds-evens is maximization of a simple
scalar ``hat'' function,
\be\label{hat}
\max_{x_1} \min\{x_1, 1-x_1\},\qquad 0\leq x_1\leq 1,
\ee
for which any method should presumably perform well.

Meanwhile, the formulae \eqref{minOMI} and \eqref{minOMO} are multiples of each other, so that the maximin problems
for (OMI) and (OMO) are essentially identical, and should display equivalent performance, giving both a useful
numerical double check and a reduction of our studies.
Both Rock-Paper-Scisssors version reduce effectively to maximization of a ``triple hat'' or ``pyramid'' function
\be\label{rpsv}
\max_{x_1,x_2} \min\{x_1, x_2,  1-x_1-x_2\},\qquad 0\leq x_1,x_2, \quad x_1+x_2 \leq 1.
\ee
See Figs. \ref{pyrfig} and \ref{pyrfigs} for contour map and smoothing.
(Note: the map is extended for numerical purposes to the full interval $[0,1]^2$, with a penalty function
enforcing the boundary condition $x_1+x_2\leq 1$.
%defined on the simplex $0\leq x_1,x_2$, $x_1+x_2 \leq 1$. 

  \begin{figure}
        \centering
		  \includegraphics[scale=0.5]{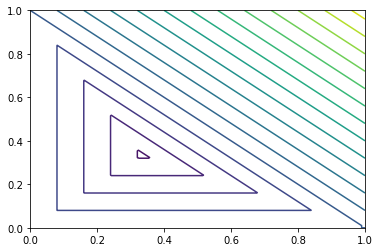}
		 	\caption{
				Pyramid function, contour map for the RPS OMI maximin.  }
	\label{pyrfig}
  \end{figure}

  \begin{figure}
        \centering
		  \includegraphics[scale=0.5]{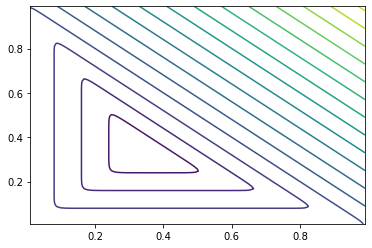}
		 	\caption{
				Smoothed contour map for the RPS OMI maximin.  }
	\label{pyrfigs}
  \end{figure}

\subsubsection{Numerical outcomes for RPS Maximin}\label{s:maximinout}
For all our RPS variants, we attempt several numerical optimization methods.  First we compare methods with soft constraints in the form of a BFGS method with penalty and methods using a hard constraint.  Additionally, we compare a nonsmoothed objective function with an $l^p$ smoothed objective function and a softmax smoothed objective function as decribed in Appendix \ref{s:numpro}.  We specifically use two examples of $\ell^p$ smoothing.  One which is very heavily smoothed, $p=2$, and one which is a closer approximation to the sharp function, $p=100$.

For the maximin RPS using either full or soft constraints, the nonsmooth method often terminates early, but largely gets quite close to the expected equilibrium strategy of $(\frac{1}{3},\frac{1}{3},\frac{1}{3})$, varying by at most about $.04$.  Meanwhile any smoothing technique will find a solution to 8 decimal points very quickly.  For this problem, though the nonsmoothed version is mostly effective, there is a clear advantage to any smoothing method.  Notably, 
since the OMI and OMO problems are equivalent for the maximin problem, we only simulate OMI.

\subsection{The Minimax problem}\label{s:minimax}
The Minimax problem \eqref{mini} does not correspond to a classical two-player game, and is generically
not equivalent to a linear programming problem, or even a convex one.
For example, even for the simplest example of odds and evens, the payoff functions 
\eqref{OMO1}, \eqref{OMO2} are nonconvex with respect to $(y,z)$.
Thus, we cannot use the simplex or interior point programming methods.
Nor does fictitious play yield a solution, as we discuss further in Section \ref{s:JFPnum} just below.
Accordingly, we are restricted to off-the-shelf routines such as (BFGS) and variants,
with no alternatives with which to compare.
Indeed, our only benchmarks here are the exact solutions found in earlier sections.

\subsubsection{Odds-Evens and Rock-Paper-Scissors}\label{s:maxminRPS}
Starting again with \eqref{RPSrptOMI}-\eqref{tildePsi}, we have that the payoff function for (OMI) is,
for both Odds-Evens and Rock-Paper-Scissors,
$$
\Psi(x,y,z)=2 y\cdot z - x\cdot(y+z)
$$
and the payoff function for (OMO), also for both, is
$$
\tilde \Psi(x,y,z):=-\Psi(x,y,z).
$$

Thus, the objective functions $\phi(y,z):=\max_x \Phi(x,y,z)$ and
$\tilde \phi(y,z):=\max_x \tilde \Phi(x,y,z)$ for the Minimax problem are given in both cases by
\be\label{mMomi}
\phi(y,z)=2y\cdot z - \min_j (y_j+z_j)
\ee
for the OMI problem and
\be\label{mMomo}
\tilde\phi(y,z)=-2y\cdot z + \max_j (y_j+z_j),
\ee
where $\psi$ and $\tilde \psi$ are to be minimized over probability vectors $y$ and $z$, with $y,z\in \R^2$
for Odds-Evens, and $y,z\in \R^3$ for Rock-Paper-Scissors.

\medskip
{\bf Odds-Evens (OMI).} Let us first consider Odds-Evens (OMI). Setting $y=y_1$, $z=z_1$, this can be written as
the minimization problem
\be\label{mMoeomi}
\min_{y,z} \alpha(y,z):=  2\big(yz+ (1-y)(1-z)\big) - \min \{ y+z, 2- y-z\},
\ee
for $0\leq y,z\leq 1$, or
$$
\alpha(y,z)=\begin{cases}
4yz -2y-2z+2  & y+z\leq 1\\
4yz-y-z& y+z\geq 1
\end{cases}
$$
with
$$
\alpha(y,z)= 1+ 4yz -2y-2z= -(2y-1)^2
$$
for $y+z\equiv 1$.

The function $\alpha$, as depicted in Figures \ref{oeomifig} below, is a sort of
nonsmooth hyperbolic paraboloid, with saddle point at the Nash equilibrium $(y,z)=(1/2,1/2)$, 
minima at $(y,z)=(0,1), (1,0)$, and a fold along the off-diagonal $y+z=1$.

  \begin{figure}
        \centering
		  \includegraphics[scale=0.5]{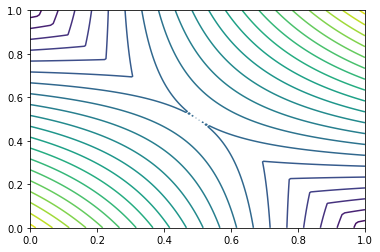}
		 	\caption{ Contour map, objective function for Odds-Evens Minimax (OMI).  }
	\label{oeomifig}
  \end{figure}

  \begin{figure}
        \centering
		  \includegraphics[scale=0.5]{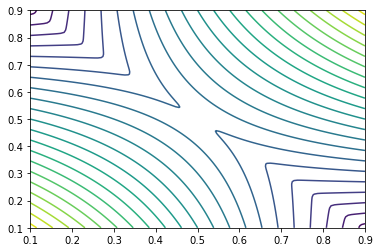}
		 	\caption{ Smoothed contour map for Odds-Evens Minimax (OMI).  }
	\label{oeomi_sfig}
  \end{figure}

  {\it Numerical outcomes for Evens-Odds OMI}: BFGS worked ``generically'' here, meaning random data converged. The convergence was always to one of the global minima in 1000 test trials.

\medskip
{\bf Odds-Evens (OMO).} We next consider Odds-Evens (OMO). Setting $y=y_1$, $z=z_1$, this can be written as
the minimization problem
\be\label{mMoeomo}
\min_{y,z} \tilde \alpha(y,z):=  -2\big(yz+ (1-y)(1-z)\big) + \max \{ y+z, 2- y-z\},
\ee
for $0\leq y,z\leq 1$, or
$$
\tilde \alpha(y,z)=\begin{cases}
-4yz +2y+2z-2  & y+z\geq 1\\
-4yz+ y+ z& y+z\leq 1
\end{cases}
$$
with
$$
\alpha(y,z)= (2y-1)^2
$$
for $y+z\equiv 1$.
Meanwhile, $\alpha(y,z)= -4y^2 + 2 y$ for $y=z$ and $0\leq y\leq 1/2$, vanishing at both endpoints.
Hence, the Nash equilibrium $(y,z)=(1/2,1/2)$ may be seen to be a global minimizer for $\alpha$,
along with $(0,0)$ and $(1,1)$, in agreement with Proposition \ref{omoprop}.
There are in addition smooth saddle points at $(1/4,1/4)$ and $(3/4,3/4)$.
See Figure \ref{oeomofig}.

  \begin{figure}
        \centering
		  \includegraphics[scale=0.5]{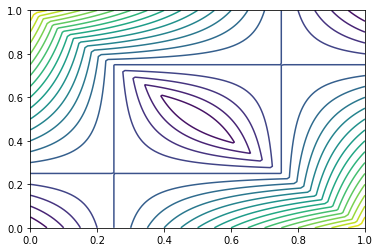}
		 	\caption{ Contour map, objective function for Odds-Evens Minimax (OMO).  }
	\label{oeomofig}
  \end{figure}

  \begin{figure}
        \centering
		  \includegraphics[scale=0.5]{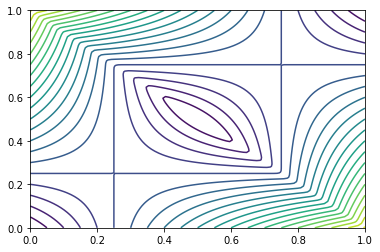}
		 	\caption{ Smoothed contour map for Odds-Evens Minimax (OMO).  }
	\label{oeomo_sfig}
  \end{figure}

  {\it Numerical outcomes for Evens-Odds OMO}: BFGS worked ``generically'' here, meaning random data converged,
  but data starting at either saddle points or corner points $(1/4,3/4)$ gave abnormal termination error, as
  did termination at Nash equilibrium. Smoothing fixed these issues in all trials.
  One can compute approximately by eye from Fig. \ref{oeomofig} the likelihood of terminating at Nash equilibrium
  vs. corner point by adding the areas of the basins of attraction, 
  slightly less than $2 (1/4)^2 + 2(1/4\times 3/4)= 2(1/4)= 1/2$.  This pattern is also observed in numerics.

 \medskip

{\bf Rock-Paper-Scissors (OMI) and (OMO).} For the Rock-Paper-Scissors Minimax, the problem is framed
in a subdomain of $\R^4$ consisting of the product of two simplices, and therefore the simple graphical
descriptions of the Evens-Odds case are not available.
However, as described in Propositions \ref{RPSOMIprop} and \ref{RPSOMOprop}, the rough features
are for (OMI) that global minimizers appear at
$y=(1,0,0),z=(0,1/2,1/2)$, $y=(0,1,0),z=(1/2,0,1/2)$, and $y=(0,0,1),z=(1/2,1/2,0)$,
with additional local minimizers
at $y=(0,1/3,2/3),z=(2/3,1/3,0)$; $y=(2/3,1/3,0),z=(0,1/3,2/3)$; $y=(1/3,0, 2/3),z=(1/3,2/3,0)$; 
$y=(1/3,2/3, 0),z=(1/3,0, 2/3)$; $y=(0,2/3, 1/3),z=(2/3, 0, 1/3)$; and $y=(0,1/3,2/3),z=(2/3, 1/3,0)$,
and a nonsmooth saddle at the Nash equilibrium $y=z=(1/3,1/3,1/3)$,
while for OMO that only global minimizers appear, 
at $y=(1/3,1/3,1/3),z=(1/3,1/3,1/3)$, $y=(1,0,0),z=(1,0,0)$, $y=(0,1,0),z=(0,1,0)$, and $y=(0,0,1),z=(0,0,1)$.
In the first case, global minimizers are mixed strategy type, and in the second all pure strategy type.
The second case features interior saddle points in the smooth portion of the domain.
The first has no associated saddle points in the smooth part of the domain, but has a
nonsmooth saddle point at the Nash equlibrium.
%END CHANGED

\subsubsection{Numerical outcomes for RPS Minimax}\label{s:minimaxout}
We run the same tests as for the maximin problem, however in this setting the OMI and OMO problems are treated separately.

For the OMI problem with full constraints, the nonsmoothed method fails about 25\% of the time.  $\ell^2$ smoothing performs very consistently, reaching a functionally exact (i.e., 7-8 digits accuracy)
solution in all 20 of our trials.  $\ell^{100}$ smoothing however failed in 15\% of our trials.   There is additionally a local min in this problem, specifically the strategy pair (1/3,2/3,0),(1/3,0,2/3).  The nonsmoothed method reached this pair in 15\% of trials, the $l^2$ smoothed problem avoided it entirely, the $\ell^{100}$ smoothing reached it 45\% of the time, and the softmax smoothed reached it 55\% of the time.  With soft constraints, the nonsmoothed method performed very badly.  While there were no outright failures, the method consistently failed to get close to the true min, or even the local min.  The $\ell^2$ smoothing also performed very badly, while the $\ell^{100}$ smoothing performed relatively well, reaching a local min in 50\% of trials.  The softmax performed similarly to the $l^{100}$ smoothing.  All of these methods consistently terminated too early, even with very small values for the stopping condition.  For this problem, hard constraints drastically outperformed the penalty method.

Meanwhile for OMO with full constraints, the nonsmoothed method failed in 5\% of our 20 trials.  In each other case, the method successfully converged to a local min.  In 80\% of trials, the method successfully converged to the true global min.  In the other 15\% of trials, the method converged to the local min of the form (1/3,2/3,0),(1/3,0,2/3).  With $l^2$ smoothing, the process converged to the true global min in all trials.  At $\ell^{100}$ smoothing, the errors were reintroduced.  The method failed in 5\% of trials, reaching the local min in 50\% of trials, and reached the true min in 45\% of trials.  The softmax method never failed, but reached the local min in 55\% of trials while reaching the global min in only 45\% of trials.  With soft constraints, the nonsmoothed method reached global min in all trials.  The $\ell^2$ smoothing reintroduced the local min in 40\% of trials, however all other smoothing methods also had a 100\% success rate.

\br\label{nolocrmk}
We conjecture that observed failure to find local minima for smoothed methods may be
a result of oversmoothing eliminating local minima.
\er

\subsection{``Joint'' fictitious play}\label{s:JFPnum}
Finally, we make a few comments regarding fictitious play, and an interesting ``joint'' variant for
coalitions.
Classical fictitious play, for any number of players, if it converges, must converge to a set of mixed strategies
that for each player $i\in \{1, \dots, n\}$ 
is optimal against the limiting strategies of the remaining players $j\neq i$.
But, this is just an alternative formulation of the definition of Nash equilibrium. Hence,

$\bullet$ {\it Fictitious play, if it converges, must converge to a Nash equilibrium.}

As, even for three player odds and evens, optimal asynchronous coalition strategies are not necessarily
Nash equilibria, this shows that fictitious play in general cannot be used to determine $V_A$.
An appealing idea is to modify the $n$-player symmetric game by pooling the winnings of players 2-$n$,
then run a standard $n$-player fictitious play algorithm on the modified game. We will call the resulting
algorithm ``joint'' fictitious play, and the resulting Nash equilibria ``joint'' Nash equilibria.

It is clear that a Nash equilibrium need not be a joint Nash equilibrium, and vice versa.
Thus, these notions do in general capture different types of information.
However, 

$\bullet$ {\it A symmetric Nash equilibrium of a zero-sum symmetric game is also a joint Nash equilibrium.}

\noindent 
For, if a change from Nash equilibrium of (without loss of generality) player 2 penalizes player 2, then it
benefits equally player 1 and players 3-n, in particular benefitting player 1. Thus, it penalizes the pooled
winnings of players 2-n.

\subsubsection{Numerical Results for Joint Fictitious Play.}
For the Evens-Odds problem, JFP seems to converge in all test cases with randomly chosen initial data.  In the case of OMI, joint ficitious play finds the Nash equilibrium in about 52\% of trials and the global minima in all other of the 1000 trials.  For OMO, joint fictitious play converged to the global minimizer (which is the same as the Nash Equilibrium for this problem) in 100\% of 1000 trials.  In each trial, 1000 iterations of the FP algorithm were performed.

For the RPS problem, JFP also seems to converge in all test cases.  For OMI, JFP finds the Nash Equilibrium in about 56\% of 1000 trials and the global minimum in all others.  For OMO, all 1000 trials converged to a 'corner' global minimizer (which are also Nash Equilibria).

\section{Experiments with Random Games I: efficiency}\label{s:numeff}
In this section we aim to explore the efficacy of different numerical methods for computing the minimax of a game as we increase the dimensions of the matrix (i.e. we increase the number of strategies available to each player).  To test this, we use randomly generated $N\times N$ matrices with entries chosen uniformly between -1 and 1,
simulating a random 2-player game. Recall that the minimax problem can even for multi-player games we formulated
as a 2-player game between two coalitions; hence, there is no loss of generality in framing the problem this way.

We compare four different methods with various parameters: Fictitious Play, Nonsmoothed minimization, $\ell^p$ Smoothed minimization, and Softmax smoothed minimization.  We additionally compare each of the minimization methods with BFGS type methods using inequality constraints and SLSQP type methods using equality constraints.  A detailed description of each of these methods can be found in Appendix \ref{s:numpro}.  

For each of $\ell^p$ smoothed and softmax, we perform runs with several parameter values.  For $\ell^p$, we run $p=1,10,50,100,200,300,400,500$.  For the softmax method, we run with $\epsilon=1,.5,.25,.0125,1e-4,1e-5,1e-6,1e-7,1e-8$.  In addition, we run Fictitious play with iterations equal to $3000,5000,12000,35000,$ $80000,1000000,4000000,12000000$. 

We measure accuracy by running each procedure twice (except Fictitious Play), once on the original matrix and once on the negative transpose.  If perfectly accurate, these two experiments will produce the same value.  Thus we use the difference between the two values, the 'value gap' as a benchmark for accuracy.  Fictitious Play has known guaranteed accuracy for a high enough number of iterations, this is discussed more in Appendix \ref{s:numpro} and Appendix \ref{s:Tables1}.

The detailed results can be found in Appendix \ref{s:Tables1}, however we summarize them here.  First we analyze the efficacy of the various smoothing methods.  For very low $N$, all methods are comparable.  However, as soon as $N=8$ the nonsmoothed method begins to drastically underperforms both smoothing methods in accuracy.  Both $\ell^p$ and softmax methods tend to perform better with parameter values smoothness, but not so high as to reintroduce the issues of the nonsmoothed methods.  At very high $N$, softmax begins to have an edge in terms of computation time, while accuracies stay comperable.

We also analyze the difference between the BFGS and SLSQP methods.  Even at low $N$, we see that the BFGS method has a clear speed advantage of about 2 orders of magnitude, while SLSQP has an accuracy advantage of about 7 orders of magnitude in the case of $N=2$.  By $N=256$, the speed advantage of BFGS reamins is closer to 3 orders of magnitude, while the accuracy advantage of SLSQP has fallen to about 1 order of magnitude.  This trend is also true for the intermediary values of $N$.  This suggests that as $N$ grows very large, there will be a point where BFGS performs unambiguously better.  However in this range of $N$, there is a reasonable tradeoff between accuracy and time.

\section{Experiments with Random Games II: $V_A$ vs. $V_N$, $V_S$}\label{s:numV}
In this section, we investigate random 3-player $N\times N\times N$ symmetric games for various small to medium sized values of $N$, comparing the efficiency, accuracy, and reliability of several different methods.  
In contrast to the previous section, our particular interest here is on {\it nonconvex} effects, and determination of
the asynchronous coalition value $V_A$, along with its relation to the synchronous and Nash values $V_S$
and $V_N$.

To simplify comparisons, we restrict to {\it symmetric} games, guaranteeing that $V_N=0$.
This is done by randomly populating a matrix with entries chosen uniformly in $[0,1]$ obeying the symmetry rules  
\ba\label{symrule}
P_{ijk}&=P_{ikj},\\
P_{ijk} &+ P_{jik}+P_{kij} = 0
\ea
where $P_{ijk}$ is return to player one, hence also:
\ba\label{rsymrule}
P_{iii}&=0,\\
P_{ijj}&=-2 P_{jji},\\
P_{iik}&=-(1/2)P_{kii},\\
P_{iji}&=P_{iij}= -(1/2)P_{jii}.\\
\ea

The methods we test include a BFGS approach with softmax smoothing and inequality constraints, a SLSQP approach with hard constraints and softmax smoothing, and a Fictitious Play based method with slight variations for $V_A$ and $V_S$.  For $V_A$ we use the Joint Fictitious Play Method described earlier, where players 2 and 3 minimize player 1's payoff independently.  For $V_S$ Fictitious Play, we join the strategies of players 2 and 3 into one synchronous strategy.

To test the convergence of each method, we compare between the predicted value of each method and the value of the actual best player 1 response to the strategy found by this method.  If the method has properly converged, then these values should be equal.  
This leads to the results in the following figures, which include the results for 10 trials on the same tensor.  Each figure uses a different payoff tensor, corresponding to a different randomly chosen symmetric game.

\begin{figure}
        \centering
		  \includegraphics[scale=0.4]{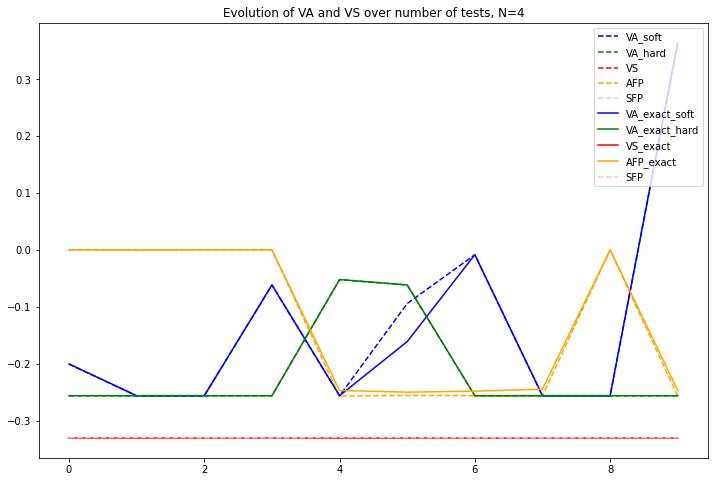}
		 	\caption{The case N=4.  We perform 10 trials of each method on a randomly generated tensor according to the symmetry rules above.  We first observe the remarkable consistency of the red and pink lines representing the synchronous problem.  Beyond this all methods perform reasonably, with SLSQP being the most consistent.  BFGS has one catastrophic failure but is generally around 0 or the true value.  Fictitious play does not quite converge for many of the trials but often comes close.  We first see a reocurring behavior of Fictitious Play, the tendency of it to find value 0.  We also note the obvious visible gap between predicted $V_A$ and $V_S$.}
  \end{figure}

  \begin{figure}
        \centering
		  \includegraphics[scale=0.4]{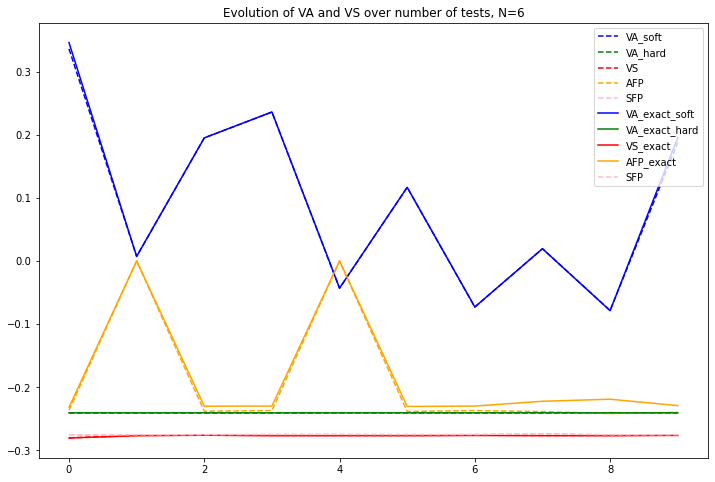}
		 	\caption{The case N=6.  Again the green line, SLSQP, shows the best performance for $V_A$ while all methods succeed for $V_S$.  The BFGS method demonstrates catastrophic failure in this case, due to failed linesearch errors, indicating a mismatched gradient.  AFP gets close to the value but does not converge well, still exhibiting the tendency to find the zero value.}
  \end{figure}

  \begin{figure}
        \centering
		  \includegraphics[scale=0.4]{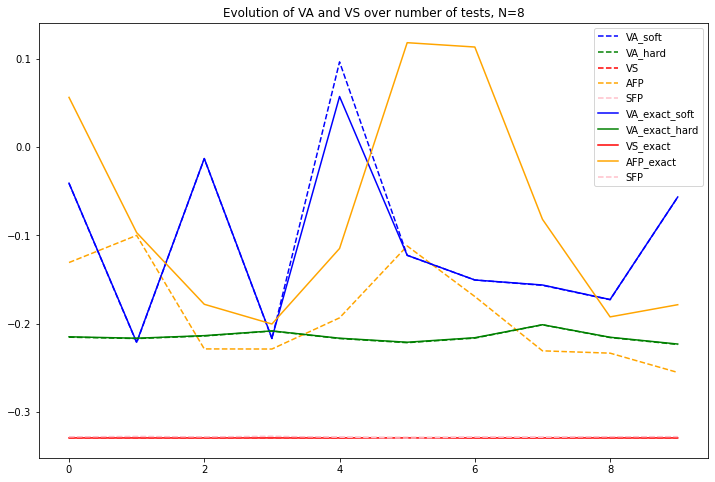}
		 	\caption{The case N=8.  Here SLSQP again performs the best.  BFGS finds the equilibrium several times and converges fairly well, but still often has linesearch errors creating premature termination.  We can clearly see the fictitious play does not converge.  In fact, it is so far from convergence that it often finds values below the true value (indicating it has found strategies which are not close to a best response pair).}
  \end{figure}

Worth noting immediately is that $V_S$ is easy to compute for all tested methods.  This is expected due to the much simpler nature of the problem.  For this reason we leave $V_S$ out of further experiments.  However, the methods struggle more to predict $V_A$.  The best performer is SLSQP. This method is very consistent, passes all tests for convergence, and predicts the minimum of the results.  The BFGS on the other hand does not work particularly well, especially for the lower $N$ values.  This is somewhat expected as BFGS is primarily meant to be applied to high dimensional problems.  We do notice that BFGS often terminates due to 'abnormal termination in linesearch,' which we will discuss more later.  The Asynchronous Fictitious Play works very well for low $N$, however even with $N$ as high as 8 it struggles to converge even when given a considerably longer runtime that the BFGS-type methods (~100 times longer runtime or more).  For Fictitious Play to be viable in higher dimensions, some improvement will be necessary.  Worth noting is the surprise that the Asynchronous Fictitious Play seems to work at all.  Unlike regular fictitious play, here we have no guarantees at all about performance.  We also performed tests on tensors without the previously described symmetry properties, and very similar results were found, though we do not display them here.

We noticed before that the BFGS method tends to fail due to a 'abnormal termination in linesearch' error, which heuristically indicates that the computed gradient does not match the true optimization landscape.  In other words, that the function is not sufficiently smoothed and that we are near a fold in the function.  To adjust for this, we add an adaptive version of BFGS.  In this method, if one of these errors is reached, we increase the $\epsilon$ parameter of the interior max, thus further smoothing the function.  This is done until the process can continue without error or the $\epsilon$ reaches a maximum threshold.  Once the process is able to continue, which indicates that the procedure has moved beyond the fold, the $\epsilon$ is reset to its original more accurate value.  This amounts to performing classical BFGS, then if a linesearch error is found, increasing the smoothing until the method can continue.

We also test an improved version of Asynchronous Fictitious Play which updates the strategies at each step by the rule
\begin{equation}
	x_{n+1} = (1-\theta(n))x_n + \theta(n) BR_x(y_n, z_n)
\end{equation}
and similarly for $y$ and $z$.  Here $\theta=1/n$ corresponds to the classical method.  
This method should allow for faster convergence by allowing larger step sizes for appropriately defined $\theta(n)$.
In particular we take $\theta(n) = max(\frac{1}{n}, .001)$.  The value $.001$ can be thought of as an accuracy threshold, since we can never get a strategy more accurate than within $.001$ of the true equilibrium except by dumb luck.  The results of both improved methods are tested on random symmetric matrices with the results encapsulated by the following figures.

\begin{figure}
        \centering
		  \includegraphics[scale=0.4]{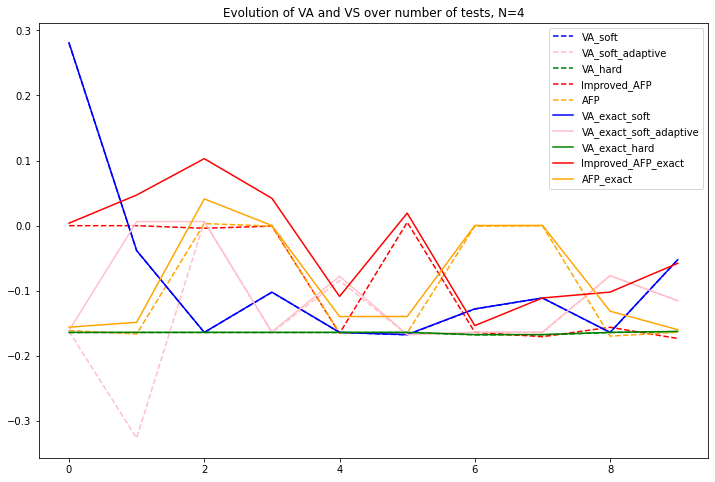}
		 	\caption{The case N=4.  We perform 10 trials of each method on a randomly generated tensor according to the symmetry rules above.  We see the green line, SLSQP, is the most consistent and best performing method.  The BFGS (blue and pink) also hits the proper value regularly (both adaptive and classical).  We can see the dotted red line (predicted improved fictitious play value) often comes very close to the true value, however since the distance from the solid red line is large, it has not converged.  The same true for classical fictitious play.  Fictitious Play also displays the typical behavior of often finding value 0.}
  \end{figure}
  
  \begin{figure}
        \centering
		  \includegraphics[scale=0.4]{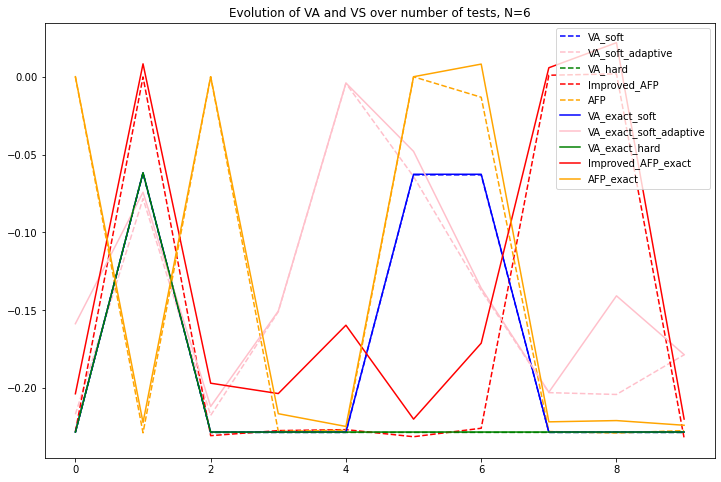}
		 	\caption{The case N=6.  Again the green line, SLSQP, shows the best performance.  The fictitious play methods perform similarly to the previous figure, bouncing between a predicted value of 0 and the true value but not converging well.}
  \end{figure}

  \begin{figure}
        \centering
		  \includegraphics[scale=0.4]{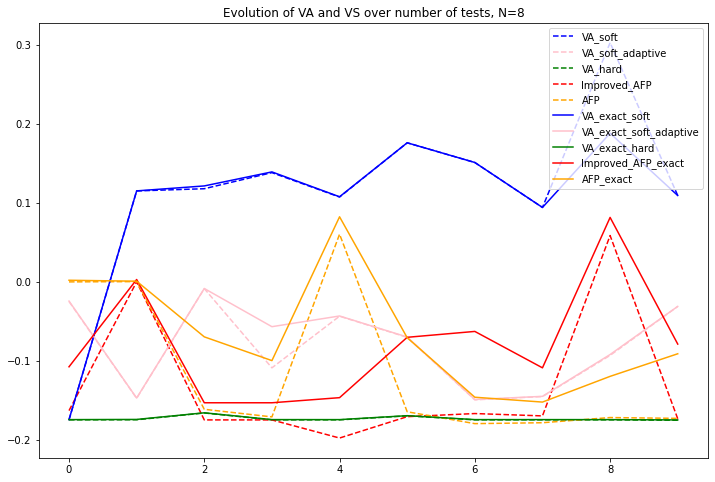}
		 	\caption{The case N=8.  Here SLSQP again performs the best.  We are entering the regime where BFGS begins to consistently terminate early, requiring the adaptive.  While the adaptive approach (pink) shows a large improvement over the traditional approach (blue), it is still not performing nearly as well as SLSQP.  We also see the improved fictitious play overtake the classical method.}
  \end{figure}

While the improvements to the methods help (especially in the cases with larger N), the SLSQP remains the best performing algorithm.  The BFGS confirms much more regularly with the adaptive method rather than the traditional, and tends to get a lower value.  The improved fictitious play shows similar improvements over the original method.  We expect that these improvements will become even more significant as $N$ is increased.  However, we emphasize that FP does 
not always converge, hence cannot be depended upon in any case as a stand-alone method.

Worth noting is that we are only able to run these simulations for relatively low values of $N$.  As shown in the two player case there is likely a point as $N$ increases where BFGS will significantly overtake SLSQP, at terms of accuracy cost per time.

\subsection{Comparison of $V_A$, $V_N$, and $V_S$: frequencies of typical gaps}\label{s:freq}
We are additionally interested in the gaps between $V_A$, $V_N=0$ and $V_S$: in particular whether
such gaps exist in generality and what are their typical sizes.  
This may be regarded as a followup for more general systems of the analyses for small systems
in Sections \ref{s:eg1} and \ref{s:eg2}.

These experiments are carried out by testing a large number of randomly chosen $N\times N\times N$ payoff tensors,
$N=4,6$, and recording the gaps $V_N- V_S=|V_S|$ and $V_N- V_A= |V_A|$.
%$N=4,6,8$, and recording the gaps $V_N- V_S=|V_S|$ and $V_N- V_A= |V_A|$.
In Figures \ref{experiments:gap4}-\ref{experiments:theta6} we display the empirical frequency distributions of the gap $V_A-V_S$ between asynchronous and synchronous coalition values (recall, as discussed in the introduction,
that this is also the gap between what coalition players can force against player 1 and what player 1 can force in asynchronous play) along with the ``relative gap''
\be\label{theta}
\theta:= V_A/V_S
\ee
measuring the relative location of $V_A$ between the two limits $V_S$ and $V_N$.
Together, these give an idea of (a) the advantage gained by synchronous coalition, and (b) the
advantage lost by asynchronicity.

\begin{figure}
\centering
%\begin{subfigure}{.5\textwidth}
  %\centering
  \includegraphics[width=.4\linewidth]{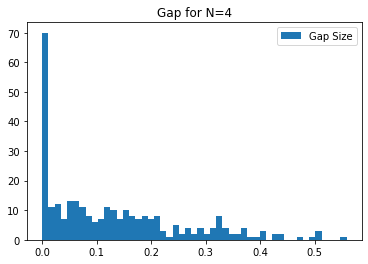}
	\caption{The gap ($V_A-V_S$) distribution for case N=4.  We notice that the gap distribution appears to decay exponentially, with a likely additional mass near 0.  This distribution has mean .12 and standard deviation .15.  We additionally observe negligible correlation with the value of $V_S$($\sim .04$).}
  \label{experiments:gap4}
\end{figure}
%\end{subfigure}
%\begin{subfigure}{.5\textwidth}
  %\centering
\begin{figure}
\centering
  \includegraphics[width=.4\linewidth]{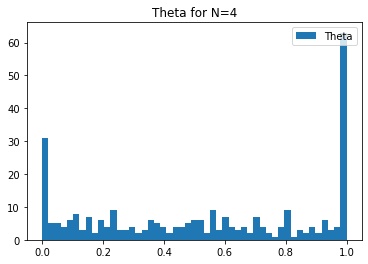}
  \caption{The $\theta$ distribution for case N=4.  We notice that this distribution seems to be roughly uniform but with concentrated masses at $0$ and $1$.  We tested several data filtering methods to see if there is consistent behavior in the formation of the masses, but did not reach any definitive conclusions.  This distribution has mean .53 and standard deviation .36.  We also note there is negligible correlation with the value of $V_S$($\sim .04$).}
  \label{experiments:theta4}
%\end{subfigure}
%\label{fig:test1}
\end{figure}

\begin{figure}
\centering
%\begin{subfigure}{.5\textwidth}
  %\centering
  \includegraphics[width=.4\linewidth]{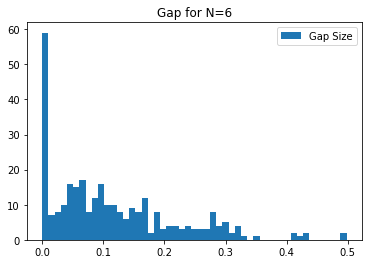}
	\caption{The gap ($V_A-V_S$) distribution for case N=6.  We notice this distribution is extremely similar to the N=4 case. This distribution has mean .11 and standard deviation .1.  We also note there is negligible correlation with the value of $V_S$($\sim .13$).}
  \label{experiments:gap6}
\end{figure}
%\end{subfigure}
%\begin{subfigure}{.5\textwidth}
  %\centering
\begin{figure}
\centering
  \includegraphics[width=.4\linewidth]{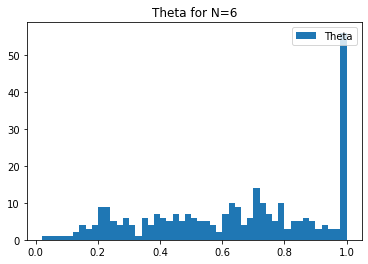}
  \caption{The $\theta$ distribution for case N=6.  We notice this distribution is extremely similar to the N=4 case.  We tested several data filtering methods to see if there is consistent behavior in the formation of the masses, but did not reach any definitive conclusions. This distribution has mean  0.63 and standard deviation .28.  We additionally observe negligible correlation with the value of $V_S$($\sim .09$).}
  \label{experiments:theta6}
%\end{subfigure}
%\label{fig:test2}
\end{figure}

It should be noted that we lack the 
%NEW
statistical 
power to draw complete conclusions about these distributions.  We only observe 300 samples, and additionally we provide no explanation for the appearance of the point masses in all histograms.

\section{Discussion and open problems}\label{s:disc}
Summing up, we have answered rather completely the problems posed in \cite{BLPWZ} regarding asynchronous
coalition games regarding whether and how often there can appear a gap between the value forcable by player 1
and the value forcable by a coalition of remaining players: that is, the situation of a nonzero gap
$$
V_A-V_S\geq 0
$$
between the asynchronous coalition value and the synchronous one. (Recall, as discussed in the introduction,
that the synchronous value $V_S$ is the best that player 1 can force in {\it either} of the synchronous or 
asynchronous games.)
In particular, we have shown that such a gap exists for the motivating example of 3-player continuous Guts
Poker, with $V_A=V_N=0$ equal to the Nash equilibrium value of zero, but $V_S$ strictly negative.
Indeed, we have gone on to give considerable further information, giving essentially a complete rigorous
solution of optimal synchronous, asynchronous, and Nash equilibrium play: quite remarkable
for a real-world version of poker that is frequently played.

Our analysis of continuous guts poker is based on special convexity properties of the payoff function, 
which at the same time suggests that this game is quite special and perhaps not a good exemplar of general behavior.
We have therefore supplemented this investigation with studies of asynchronous coalition in $N\times N\times N$
symmetric 3-player games with (i) analytical study of general games for $N= 2$ and a simple Rock Paper
Scissors game for $N=3$; and (ii) numerical study of randomly chosen games for $N=4, 6, 8, \dots$.
The latter question ties into more general issues of effective convex and nonconvex optimization for large systems.

The former investigation shows that nonzero gap is indeed a frequent occurence in asynchronous coalition play,
occurring for a wide variety of games with $N=2$ and $3$;
moreover, local minima, as might be expected, do appear, along with saddlepoints for $N\geq 3$.
The latter suggests for randomly chosen symmetric games with individual payoffs uniform in $[-1,1]$ 
that a typical gap $V_A-V_S$ for $N=4$ is $.12$
and for $N=6$ is $.25$, in both cases only mildly correlated with $V_S$.
Likewise $\theta$ in both cases is roughly one half, with weak correlation.

In the course of our numerical studies, we have compared also several off-the-shelf methods for accuracy 
and efficiency/computation time.  Our results were that SLSQP were the most dependable for finding global minima, 
but for large 2-player games (associate with convex optimization)
were prohibitively costly and at some point overtaken by smoothed BFGS methods in agreement with
conclusions of \cite{GHPS1,GHPS2}.
For 3-player games (associated with nonconvex optimization), BFGS methods often returned local rather than
global minima, making SLSQP more attractive for accuracy.  However, all methods became prohibitively expensive
for the Python-based algorithms we used, and we were unable to carry out computations for more than medium
sized systems, in particular not up to a point where we could observe crossover in efficiency from  
SLSQP to BFGS methods as in the 2-player case.

We also tested various modifications of Fictitious play with varying results. Definitely this often failed
altogether for some 3-player games, but for many was surprisingly accurate. Extremely rapid growth in computational
cost made this method noncompetitive in any case. Taken together, our results suggest that a hybrid approach
using several methods and repeated trials may be the most prudent for the moment.
However, when computationally feasible, the SLSQP method served as a dependable ``gold standard''
for our nonconvex computations for 3-player games.
(Recall, for 2-player games, Fictitious Play serves as a convenient and efficient standard \cite{BLPWZ};
Simplex or Linear Programming methods, though less convenient are also reliable \cite{D}, 
and anecdotally of similar computational cost.)

\subsection{Open problems}\label{s:open}
Regarding the motivating problem of continuous guts poker, we conjecture that $V_A=0$ also for the $n$-player
game with $(n-1)$ vs. $1$ coalition.
As for the 3-player version, this question reduces to a finite-dimensional calculus problem
that should be determinable, an interesting followup problem.
As noted in \cite{CCZ,BLPWZ}, for the actual card game of Guts Poker, played with 2-5 card hands dealt from one
or more standard 52-card decks, players' hands are not independent of one another but slightly correlated,
and so the solution of the continuous game must be modified in a somewhat irregular way to treat this case.
This problem, still open even for the (convex) 2-player game, is evidently a very interesting direction for further 
investigation, as the original motivation for all of our work.

Our results on frequency distributions of gaps for random 3-player games
can only be regarded as preliminary, given the relatively small size games that
we were able to compute, but do suggest a number of interesting questions.
For example, how do the mean values of $V_S$ and $V_A$ scale with large $N$ for symmetric
games with entries in $[-1,1]$? Or, to calibrate, what is the standard deviation of the value 
of a not necessarily random 2-player game with values in $[-1,1]$?
The small correlation observed between gap $V_A-V_S$ and synchronous value $V_S$ is puzzling,
as $V_S=0$ implies $V_A=0$: could there be an anomalous limit as $V_S\to 0$?
And, they do suggest at least that nonzero gap is a frequent occurrence for randomly chosen
3-player games, bringing to the fore the deeper game-theoretic question how to navigate
an asynchronous coalition game of this type.
Indeed, one may think of many real-world scenarios, such as ``blind'' negotiations or games where action
is too rapid or communications too slow for meaningful synchronization, where the asynchronous and not
the synchronous coalition game is relevant.

Likewise, our numerical and analytical experiments suggest that the 
asynchronous coalition problem in general possesses all the features- local as well as global minima, saddlepoints,
etc.- of general nonconvex optimization.  As such, its efficient solution belongs to the larger open
problem, applying also to machine learning, artifical intelligence, and optimal design,
of nonconvex optimization for large systems, a problem that is certainly outside the scope of the
current investigations.

Indeed, the main remaining game-theoretic open problem seems to be the same one cited in \cite{BLPWZ}.
Namely, in the scenario that $V_A=V_N=0$ but $V_S<0$- the one that we have proven above to occur for
3-player continuous Guts Poker, player 1 can guarantee at most a return of $V_S<0$. Yet, players
$2$-$3$ can guarantee a return of at most $V_A=0$.
{\it So, how can player 1 take advantage of this fact to obtain a return better than $V_S$?} 
Perhaps negotiation outside of the game?
And, how does this relate to the more familiar situation of a symmetric 3-player synchronous coalition game 
for which $V_S<0$, but $V_N=0$?

%%%%%%%%%%%%%%%%%%%  APPENDIX %%%%%%%%%%%%%%%%%%%%%

\appendix

\section{Minimax vs. critical point}\label{s:critical}
We give here a simple result used in the analysis of continuous guts.

\begin{proposition}\label{critprop}
Let $f(x,y):(\R^m \times \R^n) \to \R$ be $C^1$.
If $(x_*,y*)$ is a maximin for the problem
$$
\max_x \min_y f(x,y)
$$
and $y_*$ is the unique min of $f(x_*, \cdot)$, then $(x_*, y_*)$ is a critical point of $f$.
If $argmin f(x_*, \cdot)$ is connected and $m=1$, then there exists
	\emph{some} $y*^2\in argmin f(x_*, \cdot)$ such that $(x_*,y_*)$ is a critical point.
	Similarly, a minimax $(x_*,y_*)$ for $f(x,y):(\R^m \times \R^n) \to \R$ $C^1$ is a critical point
	if $y_*$ is the unique max of $f(x_*,y)$, or if $argmax f(x_*, \cdot)$ is connected and $m=1$.
\end{proposition}

\begin{proof} 
{\it (First assertion)} By continuity, together with uniqueness of minimum $y_*$, 
for each neighborhood $M$ of $x_*$ there is a neighborhood $N$ of $y_*$ such that 
\be \label{out}
\hbox{\rm $f(x,y)> f(x_*,y_*)$ for $x\in M$ and $y\not \in N$.}
\ee
But, also, $f(x_*,\cdot)$ is stationary at $y_*$, so that $f_y$ vanishes at $(x_*,y_*)$.  
If $f_x$ does not vanish there, then moving in a direction $h$ in $x$ for which $f$ increases at
$(x_*,y_*)$, we have by differentiability
$$
f(x_*+ th, y_*+k)=f(x_*, y_*)+  t f_x \cdot h + o(t)>0
$$
for $t\leq t_0$ sufficiently small, and arbitrary bounded directions $k$.  Taking $M$, $N$ 
to lie within the region of applicability, we have that
	\be\label{in}
	\hbox{\rm $f(x_*+t_0 h, y)> f(x_*,y_*)$ for $y\in N$}
	\ee
	Combining \eqref{out} and \eqref{in}, we find that $f(x_*+t_0 h, y)> f(x_*,y_*)$ for all $y$,
	hence, $x_*$ is not the maximin, a contradiction.  It follows that $(x_*,y_*)$ is a critical point.

{\it (Second assertion)} In this case $x$ is scalar. Taking $h=1\in \R$ in the previous argument, we find
	that $f_x$ cannot be of one sign on $argmin f(x_*,\cdot)$. For, then, by continuity, it would be
	bounded above or below, and the previous argument would again give a contradiction.
	Thus, $f_x(x_*,\cdot)$ changes sign, and so must vanish somewhere on
	$argmin f(x_*,\cdot)$, giving the result.

	{\it (Minimax)} This scenario is exactly symmetric to the first.
\end{proof}

\br\label{ceg}
If $argmin f(x_*,\cdot)$ is not assumed connected, then there are 
easy counterexamples with $argmin f(x_*,\cdot)$ consisting of two isolated points $y_*^1,y_*^2$ 
with $f_x(x_*,y_j)$ of opposite sign.
A particulary relevant one is the maximin problem $\max_{p_1^*} \min_{p_2^*,p_3^*}\alpha(p_1^*,p_2^*,p_2^*)$ 
for continuous guts, which is achieved at two different values of $(p_2,p_3)$; see Figure \ref{negativefig}.
As the function $\alpha$ has a single critical point at $p_j^*=1/\sqrt{2}$, the maximin cannot be a critical point
at either of the optimal values of $(p_2^*,p_3^*)$, as indeed is clear from the figure.
Note that this is connected with nonconvexity of $\alpha$ in $(p_2,p_3)$, as otherwise $argmin_{p_2,p_3}$
would be connected, and so, as dimension of $p_1$ is $m=1$,
there would be a critical maximin by Proposition \ref{critprop}.
\er

\br\label{relrmk}
As concerns the standard Minimax problem $\min_{(x_2,\dots, x_n)}\max_{x_1}  \Psi(x):=\sum_{i_1,\dots, i_n} x_i P_{i_1,\dots, i_n)}$ arisining the asynchronous player 2-n vs. player 1 game, $n\geq 3$, $\Psi$ is affine in $x_1$, hence 
$$
argmin \Psi(\cdot, x_{2*},\dots, x_{n*})=\R
$$
is connected but not unique. As $n-2\geq 2$, Proposition \ref{critprop} thus does not apply.
\er

\subsection{Minimax and joint Nash equilibrium}\label{s:appjNE}
It is not clear what is the relation between joint Nash equilibrium and Minimax
in general.  
For example, in RPS OMI, the global minimax $y=(1,0,0)$, $z=(0,1/2,1/2)$ and rearrangements are
joint Nash equilibria; however, the local minimax $y=(1/3,2/3,0)$, $z=(1/3,0,2/3)$ is not.

\medskip
{\bf Computations.} For RPS OMI, the payoff function is
$$
\Psi(x,y,z)= 2y\cdot z - x\cdot(y+z),
$$

The global minimax at $y_*=(1,0,0)$, $z_*=(0,1/2,1/2)$ gives
$\Psi(x, y_*,z_*)= -(x_1+x_2/2+x_3/3)$, hence the best response for player one is
$(0, \alpha, \beta)$ for arbitrary $\alpha$, $\beta$.
Varying $z=(h, 1/2-h-k, 1/2+k)$, with $h>0$ and $k$ small, gives
$$
\Psi(x_*,y_*,z)- \Psi(x_*,y_*,z_*)= (2-\alpha)h + k (\alpha-\beta).
$$
Thus, $(x_*,y_*,z_*)$ is a joint Nash equilibrium only if the above quantity is nonnegative
for all such $h,k$.  This gives, first, setting $h=0$, that $\alpha=\beta$, hence the only
candidate is $x_*=(0,1/2,1/2)$.
Then, $\Psi(x_*,y_*,z)- \Psi(x_*,y_*,z_*)= 3h/2 \geq 0 $, as required. 

A similar computation
shows that the best response for $x$ with $y_*, z_*$ fixed is $x_*$ as well, as,
setting $x=(h, 1/2-h+k, 1/2-k)$, $h>0$, we have
$$
\Psi(x,y_*,z_*)- \Psi(x_*,y_*,z_*)= -(x-x*)\cdot(1,1/2,1/2)
= -\Big( h+ (k-h)/2 -k/2 \Big)=-h/2\leq 0.
$$
Finally, settinng $y=(1-h-k h, k)$, with $h,k\geq 0$ gives
$$
\Psi(x_*,y,z_*)- \Psi(x_*,y_*,z_*)= h+k\geq 0,
$$
completing the verification that $(x_*,y_*,z_*)$ is a joint Nash equilibrium.

The local minimax at $\tilde y=(1/3,0,2/3)$, $\tilde z=(1/3,2/3,0)$ 
gives $ \Psi(x, \tilde y,\tilde z)= -4/9$ independent of $x$. 
But, then,
$$
\Psi(x,y, z)= -4/9 + x\cdot ( y-\tilde y + z-\tilde z).
$$
Setting $y=(1/3-h, 0, 2/3+h)$ and varying $h$ shows that this is a joint Nash equilibrium only
if $x_1=x_3$.
Setting $y=(1/3-h, 2/3+h,0)$ and varying $h$ shows that $x_1=x_2$, so that $\tilde x=(1/3,1/3,1/3)$
is the only candidate.
But, then, $\Psi(\tilde x, y, \tilde z)= 2y\cdot \tilde z- 2/3$, which is minimized at value $-2/3$
for $y=(0,0,1)$. Hence, this is not a joint Nash equilibrium.

\br\label{critJNErmk}
By Remark {relrmk},
for the standard Minimax problem $\min_{(x_2,\dots, x_n)}\max_{x_1}  \Psi(x):=\sum_{i_1,\dots, i_n} x_i P_{i_1,\dots, i_n)}$ arisining the asynchronous player 2-n vs. player 1 game, $n\geq 3$, $\Psi$ is affine in $x_1$,
hence an interior critical point is a joint Nash equilibrium.
However, again by Remark {relrmk}, an interior Minimax for such a game is not necessarily
a joint Nash equlibrium.
\er

\medskip
{\bf Conclusions:} It is not clear to us whether there is any useful general relation beteen joint Nash
equilibria and Minimax points. Nor is it clear that joint Fictitious play converges. However, one
may always compute them and afterward see if they correspond to Minimaxima.
This issue is one that seems interesting for further investigation.
A second important question is the issue of {\it which joint Nash equilibria are attracting under the
dynamics of joint Fictitious play.}
Both of these apparently system-by-system questions would be very interesting to treat in a general way.
%TODO: mention in open problems the questions of relation between minimax and jNE, also attraction under
%dynamics of Fictitious play.

\section{Evolution of recursive strategies}\label{s:alphabeta}
As an interesting side note, we conclude by a general discussion about recursive games and 
how optimal strategies may change over successive rounds.
For an example, consider the $2\times 2$ recursive game
\be\label{gen}
\alpha=\bp 1 & -1\\-1& 1\ep +\alpha_0 \bp 1&1\\1&1\ep,\quad
\beta=\beta_0 \bp 1 & 1 \\1 & 1\ep,
\ee
with $\alpha_0, \beta_0>0$, $\beta_0<2$.
Denote the P1 strategy by $0\leq x\leq 1$, the probability of picking strategy 1.
Then, the optimum one-shot strategy for player 1 is evidently $x=1/2$, returning values $\alpha=\alpha_0$,
$\beta=\beta_0/2$.
If repeated indefinitely, this would give recursive value 
\be\label{oneval}
V_{oneshot}:= \alpha/(1-\beta)= \alpha_0/(1- \beta_0/2).
\ee

To investigate recursive strategies, set $h=x-1/2$, taking without loss of generality (by symmetry)
$0\leq h\leq 1/2$. Then, the payoff matrix presented to P2 for a given $V\geq 0$ is
\be\label{Vgame}
x^T(\alpha+\beta V)=
\bp \alpha_0 + V\beta_0/2 + (V\beta_0+2)h  & \alpha_0  + V\beta_0/2+ (V\beta_0-2)h\ep,  
\ee
and the best P2 response evidently strategy 2, for a return to P1 of
\be
\alpha_0  + V\beta_0/2+ (V\beta_0-2)h.
\ee

We see that so long as $V\beta_0 < 2$, the best play for P1 is the one-shot strategy $h=0$, returning
$\alpha_0 + V\beta_0/2$.  If 
\be\label{crit}
V_{oneshot} \beta_0<2,
\ee
therefore, the oneshot strategy is always best, and the ultimate recursive return is $V_{oneshot}$ above.

However, if \eqref{crit} fails, or
\be\label{anticrit}
\alpha_0 \beta_0 + \beta_0>2,
\ee
then eventually the optimal strategy for \eqref{Vgame} will switch
to $h=1/2$, or $x=1$, giving a return of
$$
\alpha_0  + V\beta_0/2+ (V\beta_0-2)/2=
\alpha_0  + V\beta_0 - 1.
$$
For $\beta_0 <1$ (so that by \eqref{anticrit} $\alpha_0>1$), this payoff converges to 
$$
V_*:= \frac{\alpha_0-1}{1-\beta_0}> V_{oneshot}.
$$
For $\beta_0\geq 1$, it diverges to $+\infty$.

This seems quite illustrative for discrete games. 
For, we see that, generally, movement $h$ away from the P2 optimum gives a penalty
of {\it linear order $h$}, which balances against a benefit of order $ hV$.
Thus, if $V$ is small enough, there is no incentive to modify the one-shot game strategy. On the other hand,
if $V$ grows large enough, the optimum strategy will switch.

\subsection{A model for continuous Guts}\label{s:gutsmodel}
The situation is a bit different for continuous games like guts.  Based on our observations regarding
optimal strategies in Section \ref{s:struct}, let us consider a simplified model for synchronous coalition
continuous guts, in which we require that P2-3 use a mixture of strategies $(q_1,q_1)$ and $(0, q_2)$, weighted
as $y$, $(1-y)$.  Meanwhile we require that P1 use a pure strategy $p$.
Then, the one-shot game becomes minimizing over $(y,q_1,q_2)$ the maximum over $p$ of payoff function
\be\label{payoff}
\Psi(y,p,q_1,q_2):= y\alpha(p,q_1,q_1)+(1-y)\alpha(p,0,q_2),
\ee
or
\be\label{optprob}
\min_{(y,q_1,q_2)} \max_p \Psi(y,p,q_1,q_2),
\ee
with $\Psi$ as in \eqref{payoff}.

This is a much less computationally intensive problem than the full minimization problem,
and should be accessible by BFG, gradient descent, or simple discretization/comparison.

Note in this situation that there is a new possibility of modifying $q_j$, for which the 
resulting payoff change is $C^2$. In that case, the accounting changes, since penalties will be order $-ch^2$ 
and benefits still linear order $d Vh$.  So this suggests that optimal strategy is
$$
h\approx dV/c,
$$
leading to a recursive strategy different from the one-shot version, even for $V\ll 1$ as it is seen numerically
to be.

{\bf Conjecture:} Based on this back-of-the-envelope computation, we guess that the recursive strategy
for continuous guts does change from the one-shot optimum, but only order $V\approx 0.013$. 
If we discretize the game, then things become again a bit different: for sufficiently fine discretization,
this conclusion should persist, but for coarser grid the one-shot strategy should remain optimum
throughout play. The above model computations suggest how finely we must discretize
in order that the recursive strategy evolve,
namely, so that neighboring $q$ values are less than order $V$ apart...  At our current
discretization of $N=50$ or $N=100$, they are only $0.02$ or $0.01$ appart, right on the boundary.
And, indeed, numerical experiments show little if any evolution over different rounds of play for this
level of refinement.

\subsection{Rigorous optimum for continuous recursive Guts}\label{s:rguts}
The above questions can be answered analytically by the following extension of Proposition \ref{3br}.

\begin{proposition} \label{r3br}
For continuous recursive 3-player Guts $\alpha +V\beta$, with $V$ sufficiently small, 
the best response function is
\be\label{VR}
R^V(p_1^*) \min \{\alpha^V_a(p_1^*), \alpha^V_b(p_1^*)\},
\ee
where 
\ba\label{V3breq}
	\alpha_a^V&:= (p_3^a(V)-p_1^*)(4p_3^a(V)^2-2) + V(2-p_1^*-2 p_3^a(V) + 2p_1^*p_3^a(V)^2),\\
	\alpha_b^V &:= 2p_1^* - p_3^b(V) + p_3^b(V)^3 - 3(p_1^*)^2 p_3^b(V) + V(2-p_1^*-p_3^b(V)),
\ea
with
	\ba \label{Vopts}
	 p_3^a(V)&:=(1+V)(\sqrt{(2-V)^2(p_1^*)^2+6(1+V)}-(2-V)p_1^*)^{-1},\\
	  p_3^
	  (V)&:= \sqrt{\frac{3(p_1^*)^2+1+V}{3}}\Big),
	  \ea
	  hence the recursive synchronous value map is given by
	\be\label{Vxvalue}
	T(V)=\max_{p_1^*}\min\{\alpha_a^V(p_1^*), \alpha_b^V(p_1^*)\}.
	\ee
\end{proposition}

\begin{proof}
As in the proof of Proposition \ref{3br} given in \cite[Prop. 6.8]{CCZ}, we observe that
	for $p_1^*\not \leq p_2^*,p_3^*$, the Hessian of $\alpha(p_1^*,\cdot,\cdot)$ 
	with respect to $p_2^*,p_3^*$ has negative determinant, and so there can be no local minima.
	For $p_1^*\leq p_2^*, p_3^*$, on the other hand, the Hessian is positive definite and so
	the local minimum of $\alpha$ on this region may be found at an interior critical point.
	Setting $(\alpha + V\beta)_{p_2^*}=(\alpha + V\beta)_{p_3^*}=0$, and solving, we readily
	find that $p_2^*=p_3^*=p_3^a(V)$, with $p_3^a(V)\geq p_1^*$ so long as $p_1^*\leq 1/\sqrt{2} + O(V)$.
	Indeed, that $p_2^*=p_3^*$ at a unique critical point may be deduced already by strict convexity
	together with symmetry in $p_2^*$ and $p_3^*$.
	
	The only other candidates for minima are at boundaries $p_2^*$ or $p_3^*$ equal $0$ or $1$.
	Case by case comparison eliminates all but $\alpha(p_1^*,0, p_3^b(V))$, with
	$0=p_2^* \leq p_1^*\leq p_3^b(V)$. Here, $p_3^b(V)$ is found by setting to the derivative
	with respect to $p_3$ of $\alpha(p_1^*,0,p_3)$ and solving for $p_3$.
	Finally, we observe that for $ p_1^*\geq 1/\sqrt{2} + O(V)$, so long as $V$ is sufficiently small, 
	that $\alpha_a^V(p_1^*)> \alpha_b^V(p_1^*)$, so that formula \eqref{VR} remains correct even
	though $\alpha_a^V(p_1^*)$ is no longer a valid local minimum.\footnote{This last
	repairs a minor omission in \cite{CCZ}, where the final point was not addressed.}
\end{proof}

\br\label{Vsmall}
In practice, we may determine whether $V$ is small enough by examining the graphs of $\alpha_a^V$ and
$\alpha_b^V$ and verifying their respective positions for $p_1^*\geq p_1^{optimum}$, where
$p_1^{optimum}$ is defined as the maximum value of \eqref{Vxvalue} obtained as the intersection 
on $p_1^*\in [0,1]$ of the graphs of $\alpha_a^V$ and $\alpha_b^V$.
For continuous guts poker, the fixed point $V\approx -0.013$ of $T$ (determined in \cite{BLPWZ}) 
is more than small enough.
\er

As illustrated in Figure \ref{0.13fig}, the optimal strategy for player 1 indeed changes as $V$ goes
from $V=0$ to the approximate fixed point $V=-0.13$, going from the value 
$p_1^* \approx .6436$ for $V=0$, with $T(V) \approx -.005576$ to
$p_1^*=.6445$ for $V=-.013$, with $T(V) = -.013124$.
Hence, {\it the one-shot strategy is not optimal for the recursive synchronous coalition game.}

As noted in \cite{BLPWZ}, however, the optimal recursive strategy may be used for all rounds of the
game to achieve the optimum value. That is, though the best response evolves with time, it is not
necessary to vary the player 2-3 strategy in order to achieve the optimal result. (Examples given
in \cite{BLPWZ} show that this is necessary for some games, but not others.)

  \begin{figure}
        \centering
		  \includegraphics[scale=0.7]{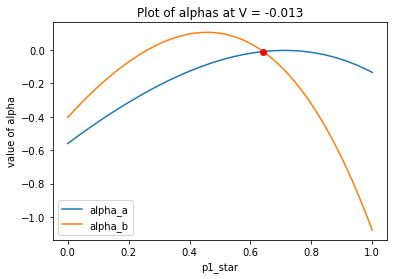}
		 	\caption{
				Blowup of graph plotting best response $\alpha^V_a$ and $\alpha_b^V$ vs. $p_1^*$,
				for $V=-0.13$.
				Maximin occurs at $p_1^*\approx 0.6445$, $\alpha\approx -0.0131$, at
				intersection of $\alpha_a^V$ and $\alpha_b^V$.
				}
	\label{0.13fig}
  \end{figure}

\subsubsection{Alternative method}\label{s:altmet}
Alternatively, defining $ p_1^m(V)$ to be the optimal strategy for player 1, defined by
\be\label{defeq}
(\alpha_a^V-\alpha_b^V)(p_1^m(V))=0,
\ee
we find by the implicit function theorem that, for $V$ sufficiently small,
\be\label{pODE}
(d/dV)p_1^m(V)= -\frac{ (d/dV)(\alpha_a^V-\alpha_b^V)} { (d/dp_1)(\alpha_a^V-\alpha_b^V)}|_{(V,p_1^m(V)}.
\ee
Here, we are recalling from our previous calculations that 
$ (d/dp_1)(\alpha_a^V-\alpha_b^V)|_{(0,p_1^m(0)} \approx 1.6 >0$.
Computing 
$$
(d/dV) (\alpha_a^V-\alpha_b^V)|_{(0,p_1^m(0)} \neq 0,
$$
we find therefore that the optimal player 1 strategy $p_1^m(V)$ indeed evolves nontrivially with $V$.

%(TODO: carry out the above computation using the definitions in \eqref{Vopts} and some annoying chain rule
%tedium, and fill in the display. OR? Other ideas how to compute this are welcome!!! -KZ)
%

\subsubsection{Conclusion}\label{s:econc}
From the above computations, we may conclude as conjectured, that the optimal strategies
for continuous guts {\it do} evolve with the round of play, and indeed from the very first round.
The change in strategy between $V=0$ and the final value $V\approx -0.013$
is $\approx .001$, that is, with respect to a change in $V$ of order $0.1$, which likewise corresponds
with conjecture.
This is convincing that evolution actually does take place, but a delicate computation that
would be difficult to confirm with numerics alone.

\section{Numerical protocol}\label{s:numpro}
1. \textbf{Numerical Simplification.} Evaluation of objective function (inner loop). In both Maximin and Minimax, this part is straightforward,
as the inner minimum (resp. maximum) is attained at pure strategy pairs (resp. single strategies), 
hence may be obtained by simple comparison as recorded in \eqref{maxi} (resp. \eqref{mini}).
This costs $N^2$ (resp. $N$) functional evaluations, of complexity $N$ (resp. $N^2$) apiece, for total
computational cost of order $N^3$, which is negligible for these problems.

2. \textbf{Optimization Methods.} When computing the minimax, we primarily use two optimization methods: Broyden-Fletcher-Goldfarb-Shanno (BFGS), and Sequential Least Squares Programming (SLSQP).  BFGS is a quasi-newton type method that uses an approximation for the hessian which improves every step.  BFGS is among the most widely used optimization method today, and has shown incredible efficiency in complex high dimensional problems.  Thus it is a natural starting point for our analysis.  SLSQP is another quasi newton method highly related to BFGS, but adjusted to handle equality constraints.  In terms of computation cost per iteration, SLSQP is many times more time expensive than BFGS.  Additional detailed references can be found REF for BFGS and REF for SLSQP.  In both cases we use the scipy python package (CITE SCIPY) as our implementation of the methods.  

3. \textbf{Smoothing.}  When computing a minimax (or maximin) problem, it is appropriate to replace the inner max (or min) with a smooth approximant in order to increase the accuracy of the minimization procedure (in fact most minimization methods we employ require $C^1$ regularity of the inner function).  We do this with two distinct methods, $\ell^p$ smoothing and softmax smoothing.  For $\ell^p$ smoothing, we use the approximate maximum function given by
\begin{equation}
    \max_i x_i \approx \left(\sum_i (x_i+1)^p \right)^{\frac{1}{p}} - 1
\end{equation}
by taking $p$ negative, we can approximate the minimum function the same way.  The shift up by 1 allows us to use negative $p$ values without worrying about dividing by $0$, since in our setting $0\leq x_i\leq1$.  It also generally helps with underflow issues since our values are between 0 and 1.  The accuracy of the approximation is increased as $|p|\to\infty$.

The softmax method meanwhile uses the approximation
\begin{equation}
	\max_i x_i \approx \sum_i \frac{x_i e^{\frac{x_i}{\epsilon}}}{\sum_j e^{\frac{x_j}{\epsilon}}}
\end{equation}
where the $\epsilon$ is a tunable parameter similar to $p$ in the $\ell^p$ smoothing example above.  In this case the approximation becomes exact as $\epsilon\to 0$.

4. \textbf{Constraints.} 
The BFGS routine is supported in SciPy with interval-type constraints $x_i\geq 0$, but
not linear inequalities $\sum_{j=1}^{N-1}\leq 1$ such as we require. We repair
this by introducing a penalty term
$-K|(\sum_{j=1}^{N-1}-1)_+|^p$ 
(resp. $+K|(\sum_{j=1}^{N-1}-1)_+|^p$), $K\gg 1$, $p\geq 1$, in the Maximin (resp. Minimax) objective function,
loosely following \cite{AY}.  We refer to these constraints as softly enforced.
Continuation-type routines in SciPi and NashPy do support linear inequalities so can be used ``as is'', we refer to these constrainst as hard enforced.  Additionally, we note that in \cite{AY}, convergence is proven for $p=2$, which is what we use in our experiments.

5. \textbf{Fictitious Play.}  Fictitious Play is a numerical algorithm for finding the Nash Equilibrium of a matrix game, or any game with a well defined best response function.  This method is initialized with two players playing a random strategy.  Then at each following iteration, each player plays the best response to the mixed strategy derived by all their opponents previous plays.  This method is not guaranteed to converge, but if it does converge to a particular strategy distribution, it is guaranteed that the distribution represents a mixed nash equilibrium.  This algorithm can also be adjusted to three player games, but does not have the same guaranteed convergence.

\section{Experiments with Random Games I Tables}\label{s:Tables1}
The following tables contain the results of the numerical experiments described in section \ref{s:numeff}.  
These experiments test three numerical methods: Fictitious Play, SLSQP, and BFGS with three different smoothing methods: no smoothing, $\ell^p$ smoothing, and softmax smoothing.  All these procedures are described in detail in Appendix \ref{s:numpro}.

For each of $\ell^p$ smoothed and softmax, we perform runs with several parameter values.  For $\ell^p$, we run $p=1,10,50,100,200,300,400,500$.  For the softmax method, we run with $\epsilon=1,.5,.25,.0125,1e-4,1e-5,1e-6,1e-7,1e-8$.  In the charts below, we show only the parameter value which produced the lowest error as described below.  In addition, we run Fictitious play with iterations equal to $3000,5000,12000,35000,$ $80000,1000000,4000000,12000000$.  For Fictitious Play, we display only the parameters which yielded similar runtimes to either $\ell^p$ or softmax.

To measure accuracy, we additionally compute the maximin by running the same procedure on the negative transpose of the payoff matrix.  If perfectly accurate, this will yield the same result as for the minimax.  Thus we use this Value Gap as a benchmark for accuracy.  
This benchmark does not work for Fictitious Play, which computes the value without reference to the minimax/maximin 
structure.
However, Fictitious Play is known analytically to converge \cite{R} for 2-player games with rate $1/n$, where $n$
is the number of iterations, hence can serve as a second benchmark in itself.
Moreover, its computational cost is relatively insensitive to the size $N$ of the game, with $N$ entering only
in the $O(N^3)$ cost  of payoff evaluations.
In the following we report the absolute value gap, but this error can also be interpreted as relative since the matrices have entries normalized to magnitude one.  Additionally we report the computational time for each method.

\begin{center}
\begin{tabular}{| c | c | c | c | c | c |} 
\hline
 \multicolumn{6}{|c|}{N=2} \\
 \hline
 Method &  Smoothing & Parameter & Value & Value Gap & Time \\ \hline
    Ficitious Play  &N/A                &    3000   &    -0.70624  &         N/A  &    0.14133\\
    SLSQP           &None               &     N/A   &    -0.70624  &  5.2425e-13  &    0.16047\\
    BFGS            &None               &     N/A   &    -0.70624  &  2.8091e-06  &   0.002279\\
    SLSQP           &$\ell^p$   &     100   &    -0.70624  &  5.2425e-13  &    0.15599\\
    BFGS            &$\ell^p$   &      10   &    -0.70624  &  2.5208e-06  &  0.0041604\\
    SLSQP           &Softmax    &   .0125   &    -0.70624  &  5.2425e-13  &    0.17692\\
    BFGS            &Softmax    &      .5   &    -0.70624  &  2.3358e-06  &  0.0021272\\
\hline
% \end{tabular}
% \end{center}

% \begin{center}
% \begin{tabular}{| c | c | c | c | c |c|} 
% \hline
 \multicolumn{6}{|c|}{N=4} \\
 \hline
 Method &  Smoothing & Parameter & Value & Value Gap & Time \\ \hline
    Fictitious Play &  N/A              &      12000   &    -0.0019507  &         N/A  &   0.44446\\
    SLSQP           & None              &        N/A   &      0.014475  &    0.040802  &    1.0221\\
    BFGS            &  None             &      N/A     &  -0.0019212    &  7.0959e-07  &  0.058216\\
    SLSQP           &$\ell^p$   &        500   &   -0.00028436  &   0.0034422  &   0.28614\\
    BFGS            &$\ell^p$   &        400   &    0.00012546  &   0.0043043  &  0.024751\\
    SLSQP           &Softmax    &       1e-5   &    -0.0019133  &  1.0477e-05  &    0.6283\\
    BFGS            &Softmax    &       1e-7   &    -0.0019213  &  4.2996e-08  &   0.06435\\
 \hline
% \end{tabular}
% \end{center}

% \begin{center}
% \begin{tabular}{| c | c | c | c | c |c|} 
% \hline
 \multicolumn{6}{|c|}{N=8} \\
 \hline
 Method &  Smoothing & Parameter & Value & Value Gap & Time \\ \hline
    Fictitious Play     &N/A     &     12000   &    -0.14168  &         N/A  &    0.4448\\
    SLSQP               &None    &       N/A   &   -0.086871  &    0.077539  &   0.91074\\
    BFGS                &None  &       N/A   &    -0.10877  &    0.090801  &   0.04243\\
    SLSQP               &$\ell^p$    &       500   &    -0.14038  &   0.0025014  &   0.22348\\
    BFGS                &$\ell^p$   &       400   &    -0.14006  &   0.0031443  &  0.034575\\
    SLSQP               &Softmax    &      1e-5   &    -0.14156  &  2.5678e-05  &    1.0319\\
    BFGS                &Softmax   &      1e-4   &    -0.14157  &  3.0448e-05  &  0.047931\\
\hline
% \end{tabular}
% \end{center}

% \begin{center}
% \begin{tabular}{| c | c | c | c | c |c|} 
% \hline
 \multicolumn{6}{|c|}{N=16} \\
 \hline
 Method &  Smoothing & Parameter & Value & Value Gap & Time \\ \hline
    Fictitious Play   &       12000  &    0.023029  &         N/A  &  0.44602\\
    SLSQP &None  &         N/A  &    0.093545  &     0.16301  &   1.0236\\
    BFGS   &  None&       N/A  &    0.080021  &     0.12615  &  0.14645\\
    SLSQP   &$\ell^p$&       500  &    0.025604  &   0.0050109  &   1.0203\\
    BFGS   & $\ell^p$&        300  &    0.027136  &   0.0084286  &  0.28557\\
    SLSQP   & Softmax&       1e-4  &    0.023402  &  6.4928e-05  &   1.4684\\
   BFGS   & Softmax&       1e-4  &    0.023864  &   0.0039738  &  0.45736\\
\hline
% \end{tabular}
% \end{center}

% \begin{center}
% \begin{tabular}{| c | c | c | c | c |c|} 
% \hline
 \multicolumn{6}{|c|}{N=32} \\
 \hline
 Method &  Smoothing & Parameter & Value & Value Gap & Time \\ \hline
    Fictitious Play &N/A        &       35000  &     -0.012266  &        N/A   &  1.3075\\
    SLSQP           &None       &         N/A  &      0.080816  &    0.17485  &   1.7749\\
    BFGS            &None       &         N/A  &      0.064499  &    0.13081  &  0.26525\\
    SLSQP           &$\ell^p$   &         500  &    -0.0089836  &  0.0053671  &   2.8667\\
    BFGS            &$\ell^p$   &         300  &    -0.0057636  &   0.010865  &    1.282\\
    SLSQP           &Softmax    &       .0125  &    -0.0070011  &  0.0095071  &   1.0497\\
    BFGS            &Softmax    &       .0125  &    -0.0067932  &   0.010705  &  0.41066\\
    \hline
\end{tabular}
\end{center}

\begin{center}
\begin{tabular}{| c | c | c | c | c |c|} 
\hline
 \multicolumn{6}{|c|}{N=64} \\
 \hline
 Method &  Smoothing & Parameter & Value & Value Gap & Time \\ \hline
    Fictitious Play &N/A        &       1000000  &   -0.0013476  &        N/A  &   38.995\\
    SLSQP           &None       &           N/A  &      0.05408  &    0.10792  &   7.0172\\
    BFGS            &None       &           N/A  &     0.067929  &    0.14123  &  0.25405\\
    SLSQP           &$\ell^p$   &           500  &    0.0011084  &  0.0055256  &    26.15\\
    BFGS            & $\ell^p$  &           200  &     0.023682  &   0.049482  &   2.6399\\
    SLSQP           &Softmax    &         .0125  &    0.0032562  &   0.010004  &   14.329\\
    BFGS            &Softmax    &         .0125  &     0.023841  &   0.052561  &   0.3972\\
\hline
% \end{tabular}
% \end{center}

% \begin{center}
% \begin{tabular}{| c | c | c | c | c |c|} 
% \hline
 \multicolumn{6}{|c|}{N=128} \\
 \hline
 Method &  Smoothing & Parameter & Value & Value Gap & Time \\ \hline
Fictitious Play&N/A  &      1000000  &   0.0041705  &        N/A  &   54.768\\
SLSQP  &None &          N/A  &    0.063219  &    0.12336  &   18.022\\
BFGS   &None&            0  &    0.083172  &    0.16983  &  0.30773\\
SLSQP  &$\ell^p$ &          500  &   0.0077296  &  0.0069334  &   265.16\\
BFGS  &$\ell^p$ &          500  &    0.054236  &    0.10957  &  0.37022\\
SLSQP  &Softmax &        .0125  &    0.010702  &   0.012289  &   54.849\\
BFGS  &Softmax &        .0125  &    0.050734  &   0.097304  &  0.52501\\
\hline
% \end{tabular}
% \end{center}

% \begin{center}
% \begin{tabular}{| c | c | c | c | c |c|} 
% \hline
 \multicolumn{6}{|c|}{N=256} \\
 \hline
 Method &  Smoothing & Parameter & Value & Value Gap & Time \\ \hline
    Fictitious Play &N/A    &     1000000 &    0.0021377  &        N/A &    62.468\\
    SLSQP  &None  &         N/A &     0.053226  &     0.1018 &    58.093\\
    BFGS   &None &         N/A &     0.092248  &    0.16708 &   0.38451\\
    SLSQP  &$\ell^p$  &         500 &    0.0056138  &  0.0067225 &       723\\
    BFGS   &$\ell^p$ &         500 &     0.047917  &   0.097521 &    1.9243\\
    SLSQP  &Softmax  &       .0125 &    0.0085595  &   0.012295 &    140.57\\
     BFGS   &Softmax &       .0125 &     0.050234  &    0.10204 &   0.57419\\
\hline
\end{tabular}
\end{center}

%%%%%%%%%%%%%%%%%%


\begin{thebibliography}{GMWZ7}

\bibitem [AY]{AY} A. Amir and A. Yassine, {\it BFGS Method for Linear Programming},
Journal of Mathematics and System Science 5 (2015) 537--543.
%doi: 10.17265/2159-5291/2015.12.006

\bibitem [BMS]{BMS} V. Boltyanski, H. Martini, and V. Soltan, {\it The Kuhn--Tucker Theorem,}
(1998). Geometric Methods and Optimization Problems. New York: Springer. pp. 78--92. ISBN 0-7923-5454-0.

\bibitem [Br]{Br} C.G. Broyden, {\it The convergence of a class of double-rank minimization algorithms}, 
Journal of the Institute of Mathematics and Its Applications 6 (1970), 76--90, doi:10.1093/imamat/6.1.76

		%introduced FP
\bibitem[B]{B} G.W.Brown, 
{\it Iterative Solutions of Games by Fictitious Play.}
In Activity Analysis of Production and Allocation, T. C. Koopmans (Ed.), New York: Wiley (1951).
		%on Bibmath.net. Archived December 26, 2008, at the Wayback Machine.

\bibitem[BLPWZ]{BLPWZ} K. Buck, J. Lee, J. Platnick, A. Wheeler, and K. Zumbrun,
	{\it Continuous Guts Poker and numerical optimization of generalized recursive games,}
		Preprint; arXiv:2208.02788.

\bibitem[CCZ]{CCZ} L. Castronova, Y. Chen, and K. Zumbrun,
	{\it Game-theoretic analysis of Guts Poker,} Preprint; arxiv:2108.06556.

\bibitem[CVX]{CVX} CVXOPT Python package, users guide: https://cvxopt.org/userguide/index.html

		%equivalence of linear programming and 2 player zero sum games, Danzig
\bibitem[D]{D} G.B. Dantzig, 
{\it A proof of the equivalence of the programming problem and the game problem.} In: Koopmans TC (ed) Activity analysis of production and allocation. Wiley, New York (1951) pp 330--335.

\bibitem[E]{E} H. Everett, {\it Recursive games,} Contributions to the theory of games, vol. 3, pp. 47--78.
Annals of Mathematics Studies, no. 39. Princeton University Press, Princeton, N. J., 1957.

\bibitem [F1]{F1} R. Fletcher, {\it Practical Methods of Optimization,} (2nd ed.), 
	New York: John Wiley \& Sons, (1987), ISBN 978-0-471-91547-8.

\bibitem [F2]{F2} R. Fletcher,
 {\it A New Approach to Variable Metric Algorithms,} Computer Journal, 13 (3) (1970), 317--322, doi:10.1093/comjnl/13.3.317.

\bibitem [GHPS1]{GHPS1} A. Gilpin, S. Hoda, J. Pen\~a, and T. Sandholm
%Andrew Gilpin,  Samid Hoda, Javier Pena, and Tuomas Sandholm1
	{\it Gradient-based Algorithms for Finding Nash Equilibria in Extensive Form Games.}
	In: Deng, X., Graham, F.C. (eds) Internet and Network Economics. WINE 2007. Lecture Notes in Computer Science, vol 4858, pp. 57-69. Springer, Berlin, Heidelberg. https://doi.org/10.1007/978-3-540-77105-0\_9

\bibitem [GHPS2]{GHPS2} A. Gilpin, S. Hoda, J. Pen\~a, and T. Sandholm
	{\it First-order algorithm with  $O(\ln(1/\epsilon))$  convergence for $\epsilon$-equilibrium in two-person zero-sum games,} Math. Program. 133 (2012), no. 1-2, 279--298.

\bibitem[Gi]{Gi} Github link.
	%, TODO.

\bibitem [G]{G} D.  Goldfarb, {\it A Family of Variable Metric Updates Derived by Variational Means,} 
	Mathematics of Computation, 24 (109) (1970), 23--26, doi:10.1090/S0025-5718-1970-0258249-6
		%Smoothing techniques for computing Nash equilibria of sequential games
%Hoda, Samid; Gilpin, Andrew; Peña, Javier; Sandholm, Tuomas
%Math. Oper. Res. 35 (2010), no. 2, 494512.
		

%\bibitem[H]{HH} Q. Huangfu, J. A. J. Hall (2018), {\it Parallelizing the dual revised simplex method},
	%TODO (journal, etc?)

%not relevant.
%\bibitem[MR]{MR} C. Minazzo and K. Rider,
	%{\it Th\'eor\`emes du Point Fixe et Applications aux Equations Diff\'erentielles,}
	%Universit\'e de Nice-Sophia Antipolis.

\bibitem[N]{N} J.F. Nash, {\it Non-Cooperative Games,}
Annals of Mathematics Second Series, Vol. 54, No. 2 (Sep., 1951), pp. 286--295.
		%(10 pages)
%Published by: Mathematics Department, Princeton University


\bibitem[NPy]{NPy} Nashpy documentation, https://nashpy.readthedocs.io/

\bibitem[Ne]{Ne} A. Nemirovsky (2004). {\it Interior point polynomial-time methods in convex programming,}
	%TODO, journal.
	%Arkadi Nemirovsky (2004). Interior point polynomial-time methods in convex programming.

%generic reference
\bibitem[O]{O} G. Owen, %Guillermo
	{\it Game theory,} Third edition. Academic Press, Inc., San Diego, CA, 1995. xii+447 pp. ISBN: 0-12-531151-6.

%\bibitem[Re]{Re} Replit depository: TODO
	%https://github.com/pangyjay/Guts.
		
	%showed conv. of FP, hamiltonian structure.
\bibitem[R]{R} J. Robinson, {\it An Iterative Method of Solving a Game,}
	Annals of Mathematics 54 (1959) 296--301.

\bibitem[SPy]{SPy} Scipy documentation, ttps://docs.scipy.org/doc/

\bibitem [S]{S} D.F. Shanno, {\it Conditioning of quasi-Newton methods for function minimization,} 
	Mathematics of Computation, 24 (111) (1970), 647--656, doi:10.1090/S0025-5718-1970-0274029-X

%\bibitem[S]{S} M. Shackleford, {\it Guts poker,} (internet article); https://wizardofodds.com/games/guts-poker/

	%introduction sg.
%\bibitem[Sh1]{Sh1} L.S. Shapley, {\it Stochastic games,} PNAS 39 (1953) no. 10, 1095--1100. 

	%original counterexample to fictitious play. 3X3, asymmetric...  nicer actually, more focused.
%\bibitem[Sh2]{Sh2} L.S. Shapley, {\it Nonconvergence of fictitious play,}
	%RAND memo RM-3026-PS (1962).

	%declassified release of counterexample to fictitious play. 3X3, asymmetric...
\bibitem[Sh3]{Sh3} L.S. Shapley, {\it Some Topics in Two-Person Games.} In Advances in Game Theory M. Dresher, L.S. Shapley, and A.W. Tucker (Eds.), Princeton: Princeton University Press (1968).


\bibitem[vN]{vN} J. Von Neumann, 
	{\it Zur Theorie der Gesellschaftsspiele,} Math. Ann. 100 (1928) 295--320.

%\bibitem[vNM]{vNM} J. Von Neumann and O. Morganstern, 
	%{\it Theory of Games and Economic Behavior,} Princeton and Woodstock: Princeton University Press (1944).

%\bibitem [WD]{WD} D. Wales, J Doyle (1998), {\it Global Optimization by Basin-Hopping and the Lowest Energy Structures of Lennard-Jones Clusters Containing up to 110 Atoms,} TODO (journal, etc?)

%\bibitem[W1]{W1} Wikipedia, https://en.wikipedia.org/wiki/Guts\_(card\_game).

%\bibitem[W2]{W2} Wikipedia, https://en.wikipedia.org/wiki/Strong\_Nash\_equilibrium.

\end{thebibliography}
\end{document}